\documentclass[journal]{IEEEtran}
\usepackage{cite}
\usepackage{amsmath,amssymb,amsfonts,amsthm}
\usepackage{graphicx}
\usepackage{multirow}
\usepackage{subfigure}
\usepackage{textcomp}
\usepackage{enumerate,paralist}
\usepackage{mathrsfs} 
\usepackage{setspace}
\usepackage{nomencl}
\usepackage{algpseudocode,algorithm}
\usepackage{threeparttable}
\usepackage{bm}
\usepackage{comment} 
\makenomenclature

\theoremstyle{definition}
\newtheorem{theorem}{\bf Theorem}
\newtheorem{lemma}{\bf Lemma}
\newtheorem{remark}{\bf Remark}

\graphicspath{{Graphics/}}	
\usepackage[colorlinks, linkcolor=black, anchorcolor=black, citecolor=black]{hyperref}
\def\BibTeX{{\rm B\kern-.05em{\sc i\kern-.025em b}\kern-.08em
		T\kern-.1667em\lower.7ex\hbox{E}\kern-.125emX}}

\begin{document}
\title{Generalized Multi-kernel Maximum Correntropy Kalman Filter for Disturbance Estimation}
\author{Shilei Li, Dawei Shi, Yunjiang Lou, Wulin Zou, Ling Shi
\thanks{Manuscript received September 10, 2022; revised June 20, 2023; accepted
September 23, 2023. The work of D. Shi was supported by National Natural Science Foundation of China under grants 62261160575 and 61973030. This work of Y. Lou was supported by the National Key Research and Development Program of China under Grant 2020YFB1313900.}
\thanks{Shilei Li, and Ling Shi are with the Department of Electronic and Computer Engineering, The Hong Kong University of Science and Technology, Hong Kong, China (e-mail: slidk@connect.ust.hk, eesling@ust.hk).} 
\thanks{Dawei Shi is with the School of Automation, Beijing Institute of Technology, China (e-mail: daweishi@bit.edu.cn).}
\thanks{Yunjiang Lou is with the State Key Laboratory of Robotics and System, School of Mechanical Engineering and Automation, Harbin Institute of Technology Shenzhen, Shenzhen 518055, China (e-mail: louyj@hit.edu.cn).}
\thanks{Wulin Zou is with Xeno Dynamics, Control Department,  Xeno Dynamics Co., Ltd, Shenzhen 518055, China (e-mail: zouwulin@xeno.com).}}
\maketitle
\begin{abstract}
Disturbance observers have been attracting continuing research efforts and are widely used in many applications. Among them, the Kalman filter-based disturbance observer is an attractive one since it estimates both the state and the disturbance simultaneously, and is optimal for a linear system with Gaussian noises. Unfortunately, The noise in the disturbance channel typically exhibits a heavy-tailed distribution because the nominal disturbance dynamics usually do not align with the practical ones. To handle this issue, we propose a generalized multi-kernel maximum correntropy Kalman filter for disturbance estimation, which is less conservative by adopting different kernel bandwidths for different channels and exhibits excellent performance both with and without external disturbance. The convergence of the fixed point iteration and the complexity of the proposed algorithm are given. Simulations on a robotic manipulator reveal that the proposed algorithm is very efficient in disturbance estimation with moderate algorithm complexity.
\end{abstract}

\section{Introduction}	
Disturbance widely exists in mechanical systems and aeronautic systems, such as industrial robotic manipulators~\cite{b1}, motion servo systems~\cite{b2}, disk drive systems~\cite{b3}, missiles~\cite{b4}, and spacecrafts~\cite{b5}. It deteriorates the control performance significantly and even induces system instability. Hence, disturbance rejection has been a key component of the controller design. 

One approach for mitigating disturbance is to use feedforward control, which can be effective when the disturbance is measurable. However, in some cases, the cost of sensors may be prohibitive or direct measurement of the disturbance may not be possible. An alternative approach is to design robust controllers. However, there is an intrinsic trade-off between the controller's robustness and its nominal performance, which is referred to as the single degree of freedom control structure~\cite{b6}. The disturbance observer (DOB) is a promising technique to address the aforementioned issues. It acts as an add-on component for the baseline controller and can increase its robustness against disturbance and recover the controller's nominal performance when disturbance disappears. Therefore, it is favored by many researchers.

Various linear disturbance observers have been designed by different researchers for different applications, which include the frequency domain-based DOB~\cite{b7}, the extended state observer (ESO) in active disturbance rejection control (ADRC)~\cite{b8,b9}, the unknown input observer (UIO) in disturbance accommodation control (DAC)~\cite{b10}, the Kalman filter-based disturbance observer (KF-DOB)~\cite{b11}; the uncertainty disturbance estimator (UDE)~\cite{b12}, and the equivalent input disturbance estimator (EID)~\cite{b13}. The frequency domain-based DOB was proposed by Ohishi et al.~\cite{b7} and the inverse of the plant model accompanied by a filter was used to estimate the lumped disturbance; the ESO was designed by Han~\cite{c9} for the purpose of estimating the lumped disturbance; the UIO was developed by Johnson~\cite{b10} which estimated the state and the disturbance by assuming that the disturbance dynamics was the \emph{a priori} knowledge; the KF-DOB~\cite{b11} also estimated the state and the disturbance simultaneously by involving the disturbance as a new state and constructing an augmented state Kalman filter; the mechanism of the UDE~\cite{b12} was quite close to the frequency domain-based DOB where a filter was utilized to make the disturbance estimation implementable; the EID ~\cite{b13} can be regarded as an alternative to the ESO by deliberately selecting the parameters. One can refer to \cite{b6} for a more comprehensive review about the DOB.      

Although many linear disturbance observers are available with different characteristics, they usually use a constant gain to update the estimate of the state or disturbance~\cite{b7,b8,b9,b10,b11,b12,b13}, which intrinsically induces a  trade-off among disturbance estimation, state estimation, and noise suppression. The constant gain cannot handle the time-varying noise characteristics effectively. For example, the KF-DOB is derived under the well-known minimum mean square error (MMSE) criterion and is the minimum variance estimator under Gaussian assumption (note that the KF gain is constant under the steady state). However, its performance degenerates significantly with heavy-tailed noise induced by disturbance. A prescription for this issue is to re-tune the noise covariance matrices at the price of sacrificing the nominal performance. However, this method usually is unsatisfactory especially when outliers are involved. Many robust techniques have been applied to KF to increase its robustness, such as the modified influence function-based KF by Masreliez et al.~\cite{c13}, Huber-based KF~\cite{c14,c15}, robust Student’s $t$-based KF~\cite{c16,c17}. Those methods improve the robustness of the KF by a bounded influence function~\cite{c13,c14,c15}, or by employing the heavy-tailed Student’s $t$-distribution~\cite{c16,c17}. However, they mainly focus on non-Gaussian noises existing in \emph{all measurements or all process channels}, rather than only existing in some specific ones.

The correntropy provides a potential tool for improving the robustness of the KF. It originates from information-theoretic learning (ITL) and has been widely used as a robust cost for machine learning~\cite{b14}, adaptive filtering~\cite{b15}, regression~\cite{b16}, and state estimation~\cite{b17,b18}. Correntropy is a local similarity measure of two random variables, which captures higher-order statistics~\cite{b18} compared with the conventional second-order error moment and hence is more suitable for applications with heavy-tailed noise. A good property of the correntropy is that the correntropy induced metric (CIM) varies from an $\ell_2$ norm to an $\ell_{0}$ norm with the growth of the error~\cite{b19}. Using this property, the maximum correntropy KF (MCKF) was derived in ~\cite{b18,c18,c19}. Its sequential form, Chandrasekhar-type recursion, and square-root form were derived in ~\cite{d19,d20,d21}. It was also extended to the nonlinear system with MC-EKF~\cite{d22}, MC-UKF~\cite{d23}, MC-GHKF~\cite{d24}, and was applied to systems with state constraints~\cite{d25,d26}, distributed state estimation~\cite{e21}, and interacting multiple model~\cite{e22}. The above correntropy-based algorithms are mainly derived under the Gaussian kernel. Actually, they can also be derived based on others kernels, e.g., the generalized Gaussian kernel and Cauchy kernel. The generalized Gaussian kernel first used by Chen et al.~\cite{b23} for adaptive filtering. After that, it had been utilized in active noise control~\cite{b24} and multiple-hypothesis detection~\cite{b25}. The Cauchy kernel was initially employed by Wang et al.~\cite{b47} for target tracking, and then it was utilized in the distributed filtering subject to cyber-attacks~\cite{b52}. Unfortunately, although these correntropy-based algorithms are robust to outliers or heavy-tailed distributions in general, they use \emph{a unified kernel bandwidth} for all channels, which are \emph{very conservative} when only some channels contain non-Gaussian noises and the others are Gaussian.

To handle this issue, in our previous works~\cite{b20,b46,b49}, we extended the definition of correntropy from random variables to random vectors and presented the multi-kernel maximum correntropy Kalman filter (MKMCKF) where the bandwidth of each channel can be tuned flexibly. With this modification, the behavior of the CIM in different channels can be designed independently. More specifically, the infinite bandwidth is applied to the Gaussian channel so that the CIM in this type of channel is an $\ell_2$ norm. As for the non-Gaussian channel, a suitable bandwidth is selected so that the CIM changes from an $\ell_2$ norm to an $\ell_0$ norm with the growth of the error. The MKMCKF is \emph{not conservative} compared with the traditional correntropy-based algorithms. However, it still has some defects: it is derived based on the Gaussian kernel, which is less powerful than the generalized Gaussian kernel; the connection between the objective function (or the kernel parameter selection) and noise distribution for a general estimation problem is vague; the detailed convergence analysis of the fixed-point algorithm in the MKMCKF is missing.

This paper aims to cope with the aforementioned problems. We first extend our previous multi-kernel correntropy under the Gaussian kernel to a generalized multi-kernel correntropy (GMKC) under the generalized Gaussian kernel and provide the corresponding generalized loss (GL) function. Then, we provide some important properties of the GMKC and build a connection between the GL and the noise distribution based on the maximum \emph{a posteriori} probability (MAP). Finally, we derive a generalized multi-kernel maximum correntropy Kalman filter (GMKMCKF) for disturbance estimation and give a sufficient condition for the convergence of the fixed-point algorithm in GMKMCKF. We also analyze the complexity and kernel parameter sensitiveness of the GMKMCKF, and compare it with the ESO~\cite{b8}, KF-DOB~\cite{b11}, MCKF~\cite{b18}, and particle filter (PF)~\cite{b34}. The major contributions of this paper lie in three aspects: Firstly, the proposed ``multi-kernel correntropy" methodology can significantly mitigate the conservatism of traditional correntropy. Secondly, we associate the GL with the noise distribution based on MAP, which illustrates the conservatism of the traditional correntropy and provides general guidance for kernel parameter selection. Thirdly, the convergence analysis of the fixed-point iteration in the GMKMCKF is given and its performance is compared with some benchmark methods. The comprehensive contributions of this paper are summarized as follows:
\begin{enumerate}[1)]
\item We find that the noise distribution in the disturbance channel is heavy-tailed and the KF \emph{cannot} serve this type of noise effectively from the modeling perspective. To cope with this issue, we propose the GMKC and GL, demonstrate their properties (\textbf{Theorem \ref{Theorem1}--\ref{Theorem5}}), and compare the GL with the least mean $p$ power (LMP) criterion.
\item We reveal that the traditional KF can be derived by an MSE criterion and is sensitive to heavy-tailed noises. To increase its robustness, we derive a novel estimator GMKMCKF by employing the GL as the cost function, which is an extension of the MCKF and MKMCKF (\textbf{Theorem \ref{Theorem7}}) but less conservative.
\item The convergence of the fixed-point algorithm in GMKMCKF is provided (\textbf{Theorem \ref{theorem9}}). Moreover, the algorithm complexity is given and the parameter sensitiveness is numerically analyzed.  Simulations on a robotic manipulator verify the effectiveness of the proposed method.
\end{enumerate}

The remainder of this paper is organized as follows. In Section II, the GMKC and GL are introduced and their properties are given. In Section III, the GMKMCKF is derived and its convergence and complexity are discussed. In Section IV, simulations are conducted to verify the effectiveness of the proposed method. In Section V, a conclusion is drawn.

\emph{Notations}: The transpose of a matrix $A$ is denoted by $A^{\prime}$. The \emph{a priori} and the \emph{a posteriori} estimate of state $x$ is denoted by $x^{-}$ and $x^{+}$, respectively. The vector with $l$ dimensions is denoted by $\mathbb{R}^{l}$ and the matrix with $m$ rows and $n$ columns is denoted by $\mathbb{R}^{m \times n}$. $X\succ 0$ ($X \succcurlyeq 0$) denotes $X$ is positive definite (semi-positive definite) matrix. The Gaussian distribution with mean $\mu$ and covariance $\Sigma$ is denoted by $\mathcal{N}(\mu,\Sigma)$. The Laplace distribution with location parameter $\mu$ and scale parameter $s$ is denoted by $\mathcal{L}(\mu,s)$. The uniform distribution with bounds $a$ and $b$ is denoted by $\mathcal{U}(a,b)$. The $p$ norm of a vector $x$ or matrix $A$ is denoted by $\|x\|_p$ or $\|A\|_p$. The $p$ power of $p$ vector norm $x$ is denoted by $\|x\|_{p}^{p}$. The expectation of a random variable $X$ is denoted by $E(X)$.
\section{Generalized Multi-kernel Correntropy}
In this section, we first provide the traditional Kalman filter. Then, we formulate an estimation problem with unknown process disturbance. Finally, we introduce the GMKC and GL and provide their properties.
\subsection{Kalman Filter}
We consider a linear time-invariant (LTI) system:
\begin{equation}
\begin{aligned}
{x}_{k+1}&={A}{x}_{k}+{w}_{k}\\
{y}_k&={C}{x}_{k}+{v}_{k}
\label{linear}
\end{aligned}
\end{equation}
where $x_k \in \mathbb{R}^{n}$ is the state, $y_{k} \in \mathbb{R}^{m}$ is the measurement, and ${w}_{k}$ and ${v}_{k}$ are Gaussian noises with $w_k \sim \mathcal{N}(0,Q_k)$ and $v_k \sim \mathcal{N}(0,R_k)$ where $Q_{k}\succcurlyeq 0$ and $R_{k} \succ 0$. The pair $(A,\sqrt{Q}_{k})$ is assumed to be controllable and $(A,C)$ is observable. The initial state $x_0 \sim \mathcal{N}(0, \prod_{0})$ is assumed to be uncorrelated with $w_k$ and $v_k$ for $k>0$. Denote the measurement set until time step $k$ as $\{{y}_k\}:=\{y_1,y_2,\ldots,y_k\}$. In KF, we have
\begin{subequations}
\abovedisplayskip=3pt
\begin{align}
\hat{x}_{k}^{-}=A\hat{x}_{k-1}^{+}
\label{kf1}
\end{align}
\begin{align}
P_{k}^{-}=AP_{k-1}^{+}A^{\prime}+Q_k
\end{align}
\begin{align}
K_{k}=P_{k}^{-}C^{\prime}(CP_{k}^{-}C^{\prime}+R_k)^{-1}
\label{gain1}
\end{align}
\begin{align}
\hat{x}_{k}^{+}=\hat{x}_{k}^{-}+K_{k}(y_{k}-C\hat{x}_{k}^{-})
\end{align}
\begin{align}
P_{k}^{+}=(I-K_{k}C)P_{k}^{-}
\label{p1}
\end{align}
\end{subequations}
where $\hat{x}_{k}^{-}$ and $\hat{x}_{k}^{+}$ is the \emph{a priori} and \emph{a posteriori} estimate of $x_k$, and $P_{k}^{-}$ and $P_{k}^{+}$ is the \emph{a priori} and \emph{a posteriori} estimate of error covariance at time step $k$, respectively.
\subsection{Problem Formulation}
In many practical applications, systems contain unknown process disturbance, i.e., 
\begin{equation}
\begin{aligned}
{x}_{k+1}&={A}{x}_{k}+{\Gamma}{d}_{k}+{w}_{x,k}\\
{y}_k&={C}{x}_{k}+{v}_{k}
\label{lineardis}
\end{aligned}
\end{equation}
where ${d}_{k} \in \mathbb{R}^{q}$ is the unknown disturbance, $\Gamma \in \mathbb{R}^{n \times q}$ map the disturbance to the state, and ${w}_{x,k}$ and ${v}_{k}$ are nominal noises. To estimate the disturbance, we treat the disturbance as a new state and construct the augmented state as $\bar{x}_{k}=[{d}_{k}^{\prime},x_{k}^{\prime}]^{\prime}$ (the aim of putting $d_k$ ahead of $x_k$ can be found in Theorem 2 in \cite{b20}). We assume that disturbance dynamics follows
\begin{equation}
d_{k+1}= d_{k}+w_{d,k}
\label{dynmodel}
\end{equation}
since we do not have the \emph{a priori} knowledge about the disturbance dynamics (the assumption $d_{k+1}= d_{k}$ is equivalent to $\dot{d}=0$ in the continuous case which is employed in many existing works \cite{b1,b51}). Then, we obtain
\begin{equation}
\begin{aligned}
{\bar{x}}_{k+1}=\bar{A} {\bar{x}}_{k} + \bar{w}_{k}\\
{y}_{k}={\bar{C}} {\bar{x}}_{k}+ \bar{v}_{k}
\label{system}
\end{aligned}
\end{equation}
with
\begin{equation}\nonumber
\begin{aligned}
\bar{A}=\left[\begin{array}{cc}
{I}&0\\
\Gamma&A
\end{array}\right],~{\bar{C}}=\left[\begin{array}{cc}
{0}&{C}\\
\end{array}\right]
\end{aligned}
\end{equation}
where $\bar{w}_{k}=[w_{d,k}^{\prime},w_{x,k}]^{\prime}$ and $\bar{v}_{k}=v_{k}$. In the conventional Kalman filter, the initial state ${x}_{0}$ is assumed to be Gaussian with $\mathcal{N}(0,\Sigma_0)$ and the noises follow
\begin{equation}\nonumber
w_{d,k}\sim \mathcal{N}(0,Q_d), w_{x,k}\sim \mathcal{N}(0,Q_x), v_{k}\sim \mathcal{N}(0,R).
\end{equation} 
Moreover, process noise $[w_{d,k}^{\prime}, w_{x,k}^{\prime}]^{\prime}$, measurement noise $v_{k}$ and initial state $x_0$ are mutually uncorrelated for $k \ge 0$. However, the Gaussian assumption of $w_{d,k}$ usually is unrealistic. In a practical application, the disturbance dynamics generally is time-varying with $d_{k+1}=f(d_{k})+w^{*}_{k}$ where $f(d_{k})$ is a time-varying nonlinear function and $w^{*}$ is the nominal disturbance noise (conventionally it is assumed to be Gaussian). We use the nominal model \eqref{dynmodel} for implementation since we are not accessible to the practical disturbance dynamics $f(d_{k})$. In this case, $w_{d,k}=f(d_{k})-d_{k}+w^{*}_{k}$ which should be heavy-tailed since it contains both the modelling mismatch $f(d_{k})-d_{k}$ and the noise $w^{*}_{k}$. A possible representation for this kind of distribution may be the $\epsilon$-contaminated mixture model. For example, we can use the uniform distribution $\mathcal{U}(a,b)$ to capture the noise induced by the modelling mismatch $f(d_{k})-d_{k}$ and employ the Gaussian distribution $\mathcal{N}(0,Q_{w})$ for the nominal noise $w^{*}_{k}$, which follows
\begin{equation}\nonumber
w_{d,k} \sim \epsilon\mathcal{U}(a,b)+(1-\epsilon)\mathcal{N}(0,Q_{w}),\quad 0<\epsilon<1
\end{equation}
where $\epsilon$ is a weight that determines the probability of a distribution occurs. Unfortunately, this mixture model cannot be approximated by a single Gaussian distribution effectively (see Fig. \ref{pdfEstimation}). This reveals that KF is not an efficient estimator for this type of noise from the perspective of noise distribution. Moreover, in some cases, the nominal noises may follow other types of distributions (e.g., the heavy-tailed distribution in~\cite{b37}, Laplace distribution in~\cite{b39}). All these factors deteriorate the estimation accuracy of the KF-DOB. 
\begin{figure}[htbp]
	\centerline{\includegraphics[width=4.5cm]{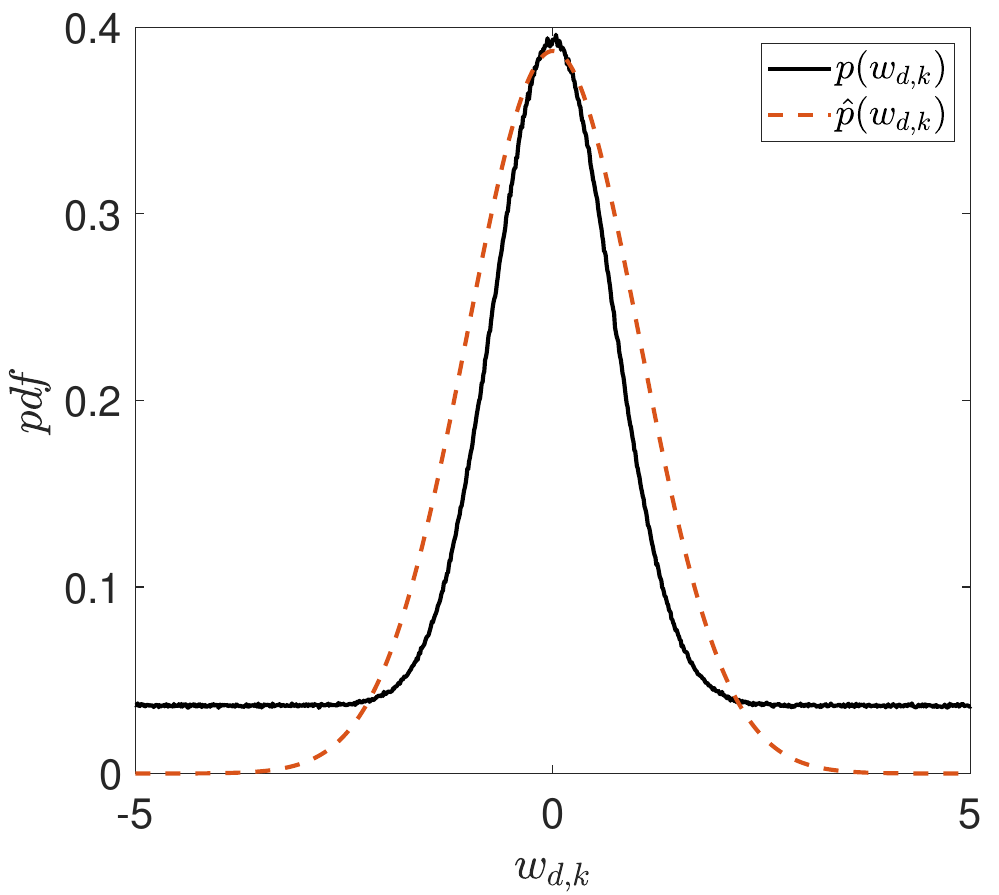}}
	\caption{Approximating a  $\epsilon$-contaminated mixture model using a Gaussian distribution. The Gaussian distribution is obtained by minimizing the mean squared error $\frac{1}{N}\sum_{k=1}^{N}\Big(p(w_{d,k})-\hat{p}(w_{d,k})\Big)^{2}$ where $p(w_{d,k}) = 0.37\mathcal{U}(-5,5)+0.63\mathcal{N}(0,0.5)$ is the target distribution and $\hat{p}(w_{d,k})$ is a Gaussian distribution to be determined. The estimated Gaussian distribution $\hat{p}(w_{d,k})$ follows $\mathcal{N}(0,1.03)$. One can see that the Gaussian distribution cannot approach a general mixture distribution effectively.}
	\label{pdfEstimation}
\end{figure}
\begin{remark}
 Although this paper focuses on process disturbance estimation, measurement disturbance actually can also be handled in a similar way by augmenting the disturbance as a new state (see Section III of \cite{b20} for details). A conventional way for the heavy-tailed distribution in the disturbed channel is to enlarge the covariance matrix $Q_d$. However, this would deteriorate its estimation performance with the disappearance of disturbance~\cite{b50}.
\end{remark}
\subsection{Generalized Multi-kernel Correntropy}
The correntropy is originally defined as a local similarity measure for two random variables $X,Y \in \mathbb{R}$ with joint distribution $F_{XY}(x,y)$  
\begin{equation}\nonumber
{C}({X},{Y})={E}[\kappa({X},{Y})]=\int \kappa(x,y) d {F}_{{X}{Y}}(x,y)
\label{correntropy}
\end{equation}  
where $\kappa(x,y)$ is a shift-invariant Mercer kernel, and $x$ and $y$ are the realizations of $X$ and $Y$. A common used kernel is the Gaussian density function with
\begin{equation}\nonumber
\kappa(x,y)=G_{\sigma}(x,y)=\exp(-\frac{e^2}{2\sigma^2})
\end{equation} 
where $e=x-y$ and $\sigma$ is the kernel bandwidth. In the case that
only $N$ samples of $x(k)$ and $y(k)$ are available and ${F}_{{X}{Y}}(x,y)$ is unknown, the correntropy can be obtained by the simple mean estimator
\begin{equation}\nonumber
{{C}}({X},{Y})=\frac{1}{N}\sum_{k=1}^{N}\kappa\big(x(k),y(k)\big)=\frac{1}{N}\sum_{k=1}^{N}G_{\sigma}\big(x(k),y(k)\big).
\label{tcorrentropy}
\end{equation}
In this paper, we adopt the generalized Gaussian density (GGD) function as the kernel
\begin{equation}
\kappa(x,y)=G_{\alpha,\beta}(x,y)=\exp({-|e/\beta|^{\alpha}})
\end{equation}
where $e=x-y$ is the error, $\alpha>0$ is the shape parameter, and $\beta>0$ is the kernel bandwidth. Under the GGD, we define the GMKC for random vectors $\mathcal{X}, \mathcal{Y} \in \mathbb{R}^{l}$ as follows (the $i$-th element of $\mathcal{X}$ and $\mathcal{Y}$ is $\mathcal{X}_{i}$ and $\mathcal{Y}_{i}$, respectively):
\begin{equation}\nonumber
\bar{C}(\mathcal{X},\mathcal{Y})=\sum_{i=1}^{l}{E}[\tilde{\kappa}_{i}(\mathcal{X}_i,\mathcal{Y}_{i})]=\sum_{i=1}^{l}\int \tilde{\kappa}_{i}({x}_{i},{y}_{i}) d {F}_{\mathcal{X}_{i}\mathcal{Y}_{i}}({x}_{i},{y}_{i})
\label{correntropyVector}
\end{equation}
with
\begin{equation}\nonumber
\tilde{\kappa}_{\alpha,\beta_i}({x}_{i},{y}_{i})=\beta_i^{\alpha}G_{\alpha,\beta_i}(x_i,y_i)=\beta_i^{\alpha}\exp({-|e_i/\beta_i|^{\alpha}})
\end{equation} 
where $x_i$ and $y_i$ are realizations of $\mathcal{X}_{i}$ and $\mathcal{Y}_{i}$, $e_i=x_i-y_i$ is the realization error, and $\beta_i$ is the $i$-th bandwidth for $\mathcal{X}_{i}$ and $\mathcal{Y}_{i}$. In a practical application, the joint distribution ${F}_{\mathcal{X}_{i}\mathcal{Y}_{i}}({x}_{i},{y}_{i})$ is not available and only $N$ samples can be obtained. In this case, we can estimate the GMKC as
\begin{equation}
\bar{{C}}(\mathcal{X},\mathcal{Y})=\sum_{i=1}^{l}\beta_i^{\alpha}{C}_{\alpha,\beta_i}(\mathcal{X}_{i},\mathcal{Y}_{i})
\label{mkcorrentropy}
\end{equation}
with
\begin{equation}
\begin{aligned}
{C}_{\alpha,\beta_i}(\mathcal{X}_{i},\mathcal{Y}_{i})&=\frac{1}{N}\sum_{k=1}^{N}G_{\alpha,\beta_i}\left(x_{i}(k),y_{i}(k)\right)
\label{correntropypair}
\end{aligned}
\end{equation}
where ${C}_{\alpha,\beta_i}(\mathcal{X}_{i},\mathcal{Y}_{i})$ is the correntropy for $\mathcal{X}_{i},\mathcal{Y}_{i}$ under $\alpha$ and $\beta_i$, and $x_i(k)$ and $y_i(k)$ is the $k$-th sample of random variables $\mathcal{X}_{i}$ and $\mathcal{Y}_{i}$, respectively. Correspondingly, the generalized loss (GL) can be defined as 
\begin{equation}
\begin{aligned}
J_{GL}(\mathcal{X},\mathcal{Y})&=\sum_{i=1}^{l}\beta_{i}^{\alpha}\left(1-{C_{\alpha,\beta_i}}(\mathcal{X}_{i},\mathcal{Y}_{i})\right).
\label{gcloss}
\end{aligned}
\end{equation}
\begin{remark}
	It is worth mentioning that the proposed {GMKC} is different from the concept in \cite{b40,b41}. The mechanism of our proposed method is to use different kernel bandwidths at different channels, while \cite{b40,b41} employ a combination of different kernels to generate a new kernel. The purpose of our method is to reject the heavy-tailed noises in the disturbing channel without sacrificing the performance of the other channels while the aim of \cite{b40,b41} is to accommodate more complex error distributions (i.e., skewed distributions, see Fig. 1 in \cite{b40}).
\end{remark}
\subsection{Properties of the Generalized Multi-kernel Correntropy}
In this section, we provide some properties of the GMKC and GL.
\begin{theorem}
	In the case of ~$0<\alpha \le 2$, the GMKC in \eqref{mkcorrentropy} can be regarded as a weighted summation of the second-order statistic in the mapped feature space.
	\label{Theorem1}
\end{theorem}
The proof of this theorem is shown in Appendix \ref{proofTheorem1}.

\begin{theorem}
	When setting ${\beta_{i}^{\alpha}} \to \infty$,
	the GL in \eqref{gcloss} becomes the expectation of $\alpha$-order absolute moments with $\lim\limits_{\beta_{i}^{\alpha} \to \infty} J_{GL}(\mathcal{X},\mathcal{Y})={E}\|\mathcal{X}-\mathcal{Y}\|_{\alpha}^{\alpha}$.
	\label{Theorem2}
\end{theorem}

The proof of this theorem is shown in \ref{proofTheorem2}. 
\begin{remark}
Theorem \ref{Theorem2} reveals that when setting all kernel parameters as $\beta_{i}^{\alpha} \to \infty$, the GL becomes the traditional least mean $p$-power (LMP) criterion with $\alpha=p$. One can refer to \cite{b30,b36} for more information about the LMP in the design of a filter.
\end{remark}
\begin{theorem}
	Denote the correntropy induced metric as $\mathrm{GCIM}(\mathcal{X},\mathcal{Y})=\left(J_{GL}(\mathcal{X},\mathcal{Y})\right)^{\frac{1}{2}}$. Then, it defines a metric in the $N$-dimensional sample vector space when $0<\alpha \le 2$.
	\label{Theorem3}
\end{theorem}
The proof is shown in Appendix ~\ref{proofTheorem3}. The contour plots of $J_{GL}(\mathcal{X},{0})^{\frac{1}{\alpha}}$ in 2D space with different shape parameters $\alpha$ and different bandwidths $\beta_i$ are shown in Fig. \ref{cim}. One can see that $J_{GL}(\mathcal{X},{0})^{\frac{1}{\alpha}}$ behaves like an $\ell_{\alpha}$ norm in the vertical direction when setting $\beta_2$ to be a big value (i.e., 100). Moreover, it changes from an $\ell_{\alpha}$ to $\ell_0$ in the horizontal direction when setting $\beta_1$ to be a relatively small value (i.e., 1). For the traditional correntropy, the contour plot is isotropic since it shares a unified bandwidth~\cite{b19}, which restrains its capability on the system that only some channels contain heavy-tailed noises. On the contrary, the contour plot of the proposed method can be anisotropic by using different bandwidths at different channels, which is very efficient when different channels contain different types of noise distributions. Another advantage of the proposed method is that the GGD is more powerful than the Gaussian density function (since it has an additional shape parameter $\alpha$) and hence can accommodate more types of noise distributions.
\begin{figure*}[htbp]
	\centering
	\subfigure[$\alpha=1$, $\beta_1=1$, $\beta_2=100$]{
		\begin{minipage}[t]{0.264\linewidth}
			\centering
			\includegraphics[width=1.0\columnwidth]{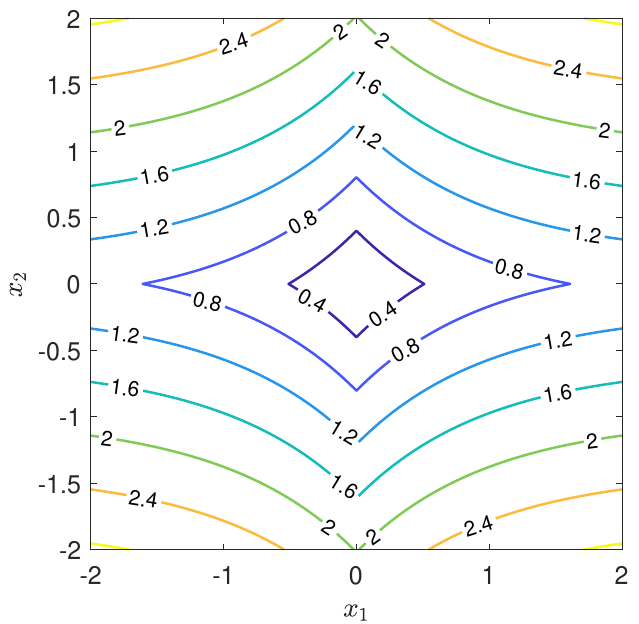}
			\label{kernels_c}
		\end{minipage}%
	}%
	\subfigure[$\alpha=2$, $\beta_1=1$, $\beta_2=100$]{
		\begin{minipage}[t]{0.260\linewidth}
			\centering
			\includegraphics[width=1.0\columnwidth]{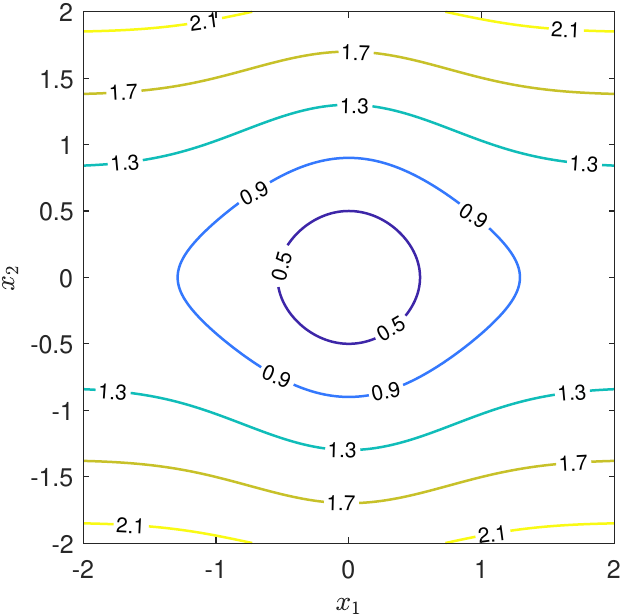}
			\label{kernels_f}
		\end{minipage}%
	}%
	\subfigure[$\alpha=4$, $\beta_1=1$, $\beta_2=100$]{
		\begin{minipage}[t]{0.262\linewidth}
			\centering
			\includegraphics[width=1.0\columnwidth]{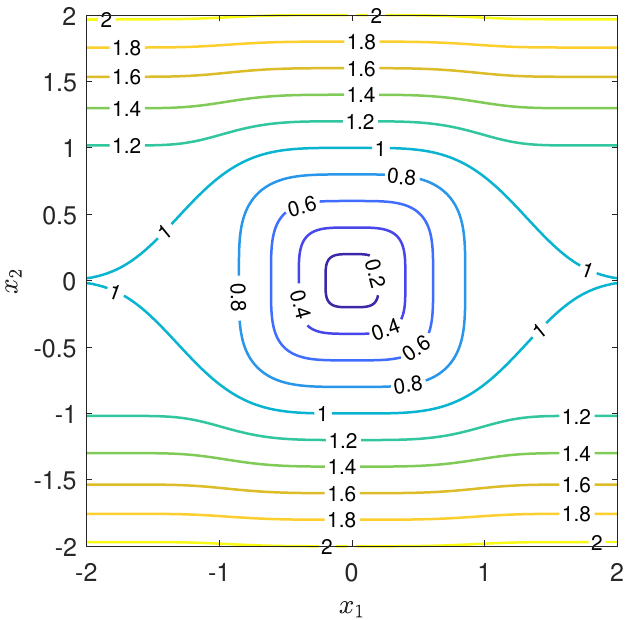}
			\label{kernels_i}
		\end{minipage}%
	}%
	\caption{Contours of $J_{GL}(\mathcal{X},{0})^{\frac{1}{\alpha}}$ in 2D space with different shape parameters and different bandwidths.}	
	\label{cim}
\end{figure*}

\subsection{Influence Function of the LMP and GL}
\label{infpec}
In many applications, we have only one measurement at each time instance. In this section, we discuss the property of the LMP and GL in this scenario, i.e., $N=1$.

The influence function measures the derivative of the loss function with respect to the error~\cite{b28,b29}, and gives a straightforward view of how errors influence the objective function. Therefore, it provides guidance for the objective function design. For the LMP criterion~\cite{b30,b36}, we have 
\begin{equation}
\begin{aligned}
J_{LMP}(e)&=\|e\|_{p}^{p}=\sum_{i=1}^{l}|e_{i}|^{p}
\label{LMP1}
\end{aligned}
\end{equation}
where $e \in \mathbb{R}^{l}$ and $e_{i}$ is the $i$-th element of $e$. Substituting \eqref{correntropypair} into \eqref{gcloss} with $N=1$, we have
\begin{equation}
\begin{aligned}
J_{GL}(e)=\sum_{i=1}^{l}\beta_{i}^{\alpha}\left(1-G_{\alpha,\beta_i}\left(e_{i}\right)\right).
\label{gcloss1}
\end{aligned}
\end{equation}
The influence functions can be obtained by calculating the gradients 
\begin{equation}
\begin{aligned}
\nabla  J_{LMP}(e)&=\frac{\partial J_{LMP}}{\partial e}=\left[\rho_{1},\rho_{2},\cdots,\rho_{l}\right]^{T}\\
\nabla  J_{GL}(e)&=\frac{\partial J_{GL}}{\partial e}=\left[\gamma_{1},\gamma_{2},\ldots,\gamma_{l}\right]^{T}
\label{infl}
\end{aligned}
\end{equation} 
with 
\begin{equation}\nonumber
\begin{aligned}
\rho_{i}&=p\frac{|e_{i}|^p}{e_{i}}, i=1,2,\ldots,l\\
\gamma_{i}&=\frac{\alpha\exp^{-\frac{|e_i|^{\alpha}}{\beta_{i}^{\alpha}}}|e_i|^{\alpha}}{e_i},i=1,2,\ldots,l.
\label{inf2}
\end{aligned}
\end{equation}
Then, we have the following two theorems.
\begin{theorem}
	The $J_{GL}(e)$ in \eqref{gcloss1} is identical to the $J_{LMP}(e)$ in \eqref{LMP1} when $\alpha=p$ and $\beta_{i}^{\alpha} \to \infty$. Moreover, in the case of ~$0<\alpha\le 1$, $J_{GL}(e)$ is concave with $e \neq 0$; in the case of~ $\alpha>1$, $J_{GL}(e)$ is convex within the region $|e_{i}|\le (\frac{\alpha-1}{\alpha})^{\frac{1}{\alpha}}\beta_{i}$.
	\label{Theorem4}
\end{theorem}
The proof of this theorem can be found in Appendix \ref{proofTheorem4}.
\begin{remark}
	In many situations, a non-convex loss function is beneficial to strengthen some particular features. For example, the conventional MSE loss gives a linear influence function (i.e., $p=2$ in \eqref{infl}), which provides each residual constant influence and hence cannot eliminate the effect of outliers (if exists). On the contrary, a redescending influence function that is induced by a non-convex loss (e.g., the GL in \eqref{gcloss1}) is preferable~\cite{b28}. Existing solutions for non-convex optimization include the fixed-point iteration~\cite{b18}, the gradient descent~\cite{b40}, and the evolutionary algorithms~\cite{c1}. 
\end{remark}
\begin{theorem}
	The GL in \eqref{gcloss1} is a differential invex function of $e$ with $\alpha>1$ and $e_{i}\le \varphi $ ($i=1,2,\cdots,l$) where $\varphi \in \mathbb{R}^{+}$ is an arbitrary positive number.
	\label{Theorem5}
\end{theorem}
The proof of this theorem is shown in Appendix \ref{proofTheorem5}. We consider the loss function of \eqref{gcloss1} in one-dimensional case for simplicity. In this case, $J_{GL}(e)=\beta^{\alpha}(1-G_{\alpha,\beta}(e))$
and $J_{LMP}(e)=|e|^{p}$. The graphs of $J_{GL}(e)$, $\nabla J_{GL}(e)$, $J_{LMP}(e)$, and $\nabla J_{LMP}(e)$ are shown in Figs. \ref{obj_b}, \ref{inf_b}, \ref{obj_c}, and \ref{inf_c}. One can see that $J_{GL}$ approaches $J_{LMP}$ when setting $\beta=100$ (see Theorem \ref{Theorem2}), and it changes from $\|e\|_{\alpha}^{\alpha}$ to $\beta^{\alpha}$ with the growth of the error when setting $\beta=1$. The influence function $\nabla J_{GL}$ goes towards zero when the error is bigger than $\frac{\alpha-1}{\alpha}\beta$ and $\alpha>1$ (see Theorem \ref{Theorem4}), and is close to $ \nabla J_{LMP}(e)$ when the error is very small, which makes the performance of GL is similar to LMP when the error is small, but is highly resistant to outliers when the error is large.
\begin{remark}
	The well-known MSE and LMP actually is a subset of the GL. The MSE-based algorithm is sensitive to outliers since its influence function grows linearly with respect to $e$ (note that $\frac{\partial |e|^2}{\partial e}=2e$). On the contrary, this effect can be mitigated by the GL by using a relatively small kernel bandwidth with $\alpha>1$ since its corresponding influence function goes towards zero with the increment of the error. Due to this property, the GL is a more attractive loss function compared with the LMP and MSE criterion.
\end{remark}
\begin{figure}[htbp]
	\centering
	\subfigure[]{
		\begin{minipage}[t]{0.48\linewidth}
			\centering
			\includegraphics[width=1.0\columnwidth]{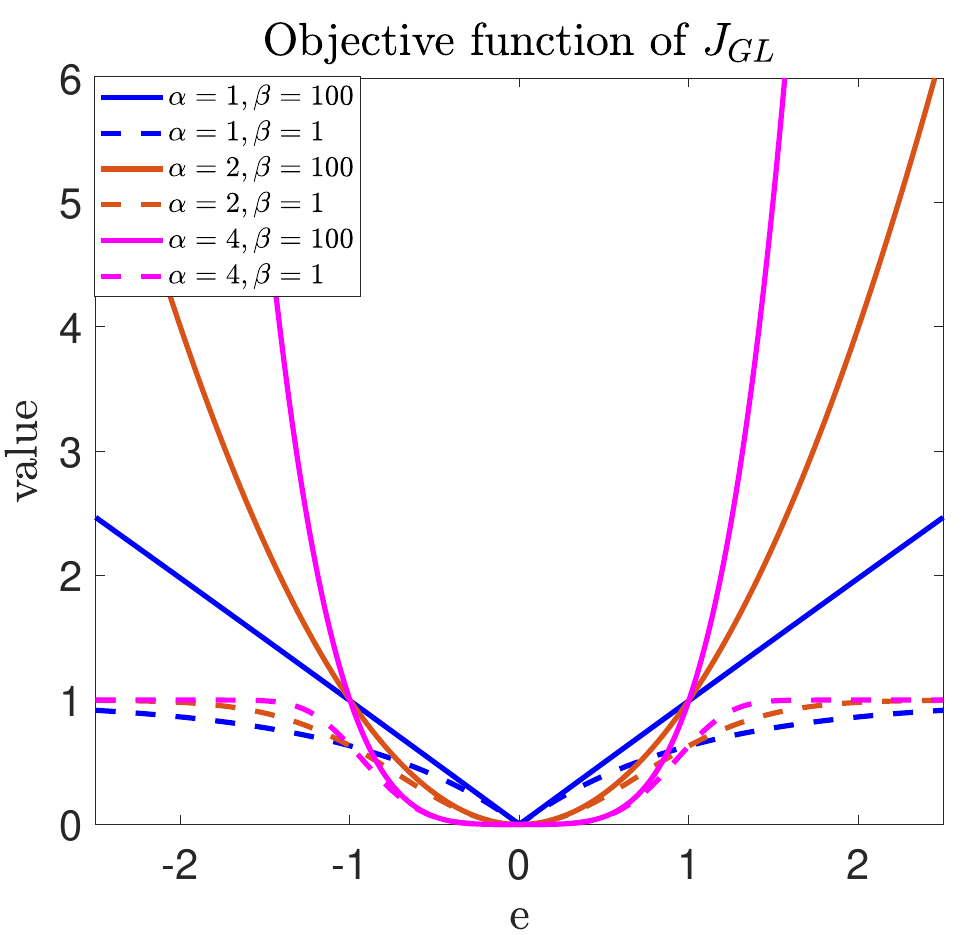}
			\label{obj_b}
		\end{minipage}%
	}%
	\subfigure[]{
		\begin{minipage}[t]{0.49\linewidth}
			\centering
			\includegraphics[width=1.0\columnwidth]{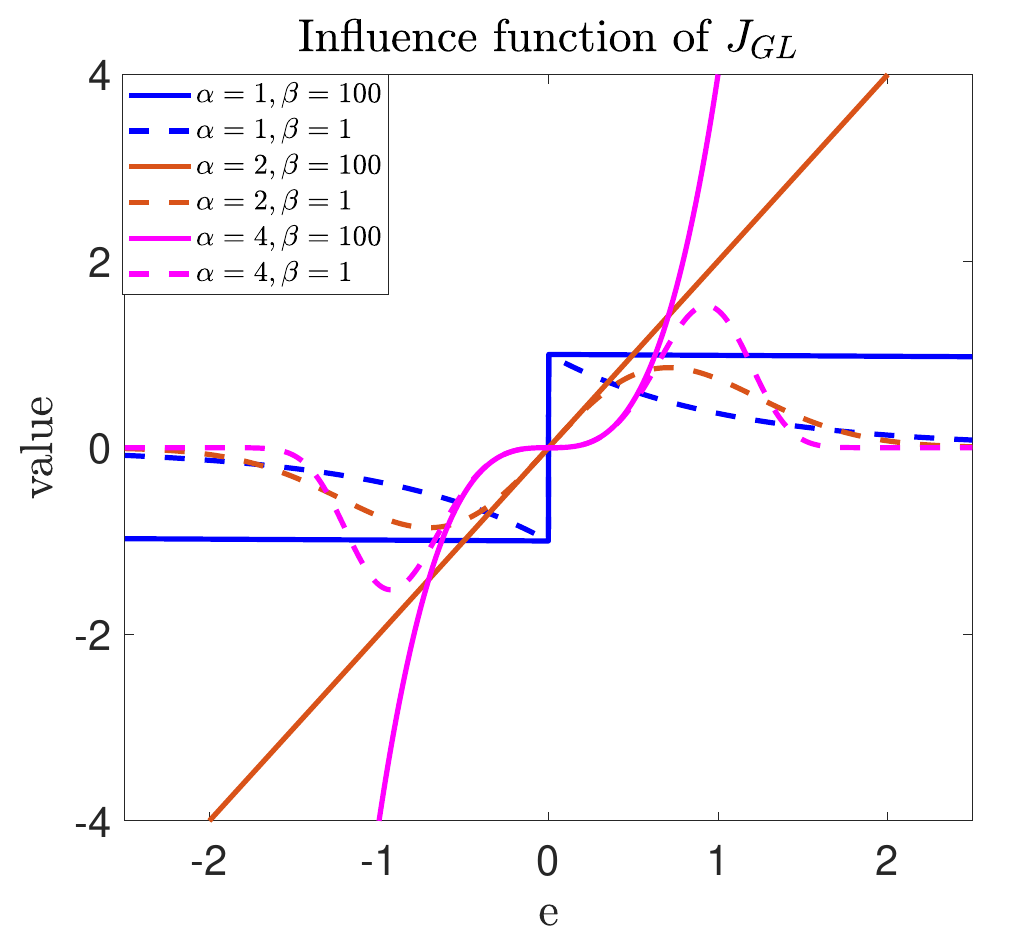}
			\label{inf_b}
		\end{minipage}%
	}%
	
	\subfigure[]{
		\begin{minipage}[t]{0.48\linewidth}
			\centering
			\includegraphics[width=1.0\columnwidth]{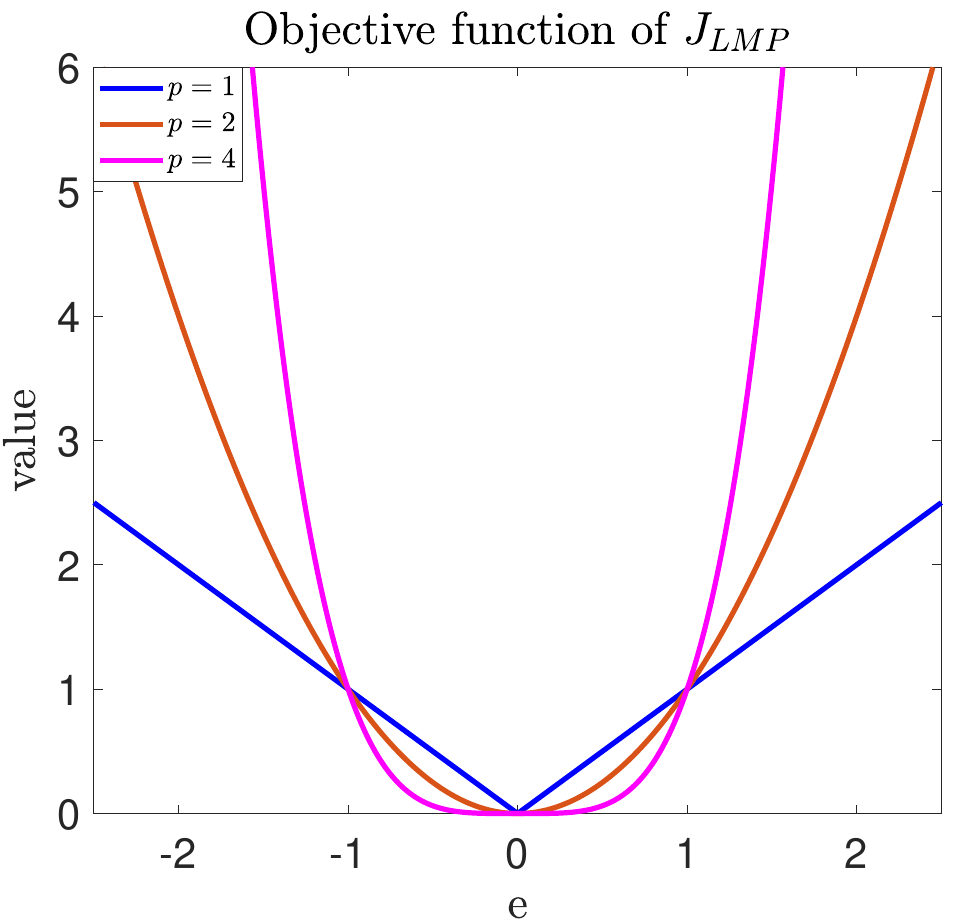}
			\label{obj_c}
		\end{minipage}%
	}%
	\subfigure[]{
		\begin{minipage}[t]{0.49\linewidth}
			\centering
			\includegraphics[width=1.0\columnwidth]{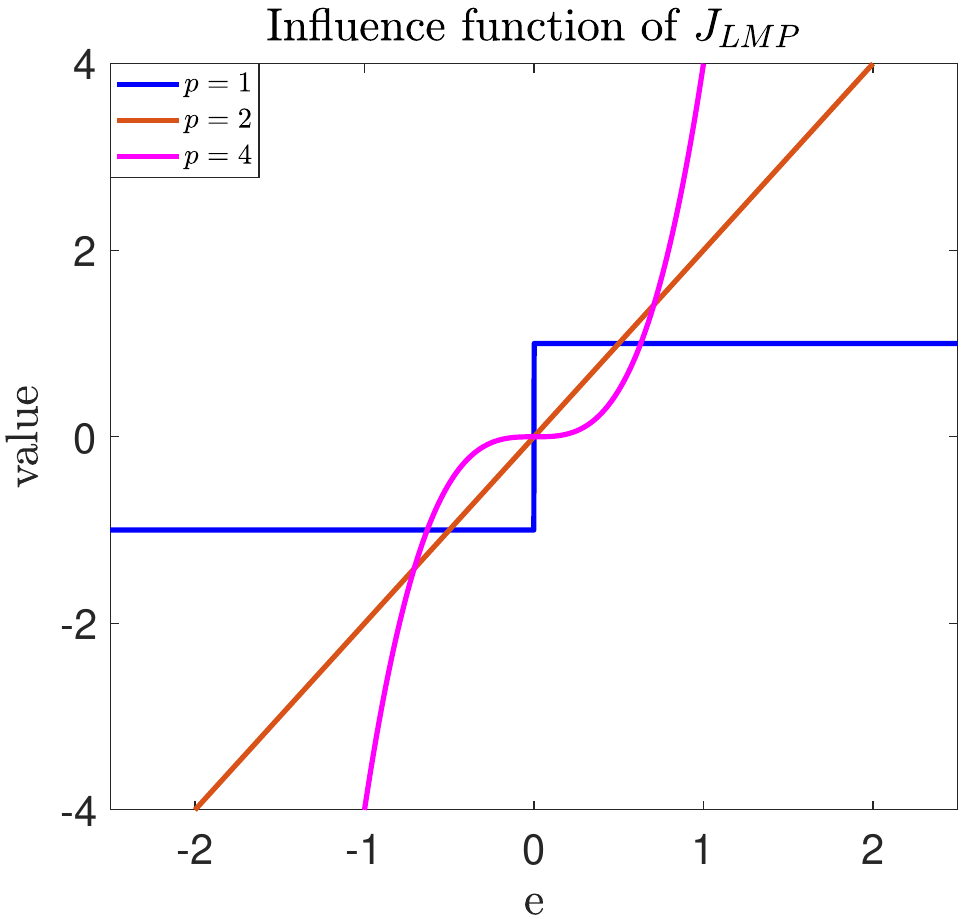}
			\label{inf_c}
		\end{minipage}%
	}%
	\caption{Objective functions and influence functions of $J_{GL}$ and $J_{LMP}$ with different $\alpha$, $\beta$ and $p$.}	
	\label{obj}
\end{figure}

\subsection{Relationship with the Kalman Filter}
\label{pdfpec}
For the linear system with Gaussian assumption described in \eqref{linear}, based on the Bayes' theorem, the \emph{a posteriori} probability of $x_k$ with measurement set $\{y_k\}$ has  
\begin{equation}
\begin{aligned}
&p(x_{k}|\{y_{k}\})=p(x_{k}|\{y_{k-1}\},y_k)=\frac{p(x_k,\{y_{k-1}\},y_k)}{p(\{y_{k-1}\},y_k)}\\
&=\frac{p(y_k|x_k,\{y_{k-1}\})p(x_k|\{y_{k-1}\})p(\{y_{k-1}\})}{p(y_k|\{y_{k-1}\})p(\{y_{k-1}\})}\\
&
\propto  p(y_k|x_k)p(x_k|\{y_{k-1}\})
\label{iterBaye}
\end{aligned}
\end{equation}
where $p(y_k|x_k)$ is the probability of $y_k$ conditional on the \emph{a priori} estimate of $x_k$, and $p(x_k|\{y_{k-1}\}$ is the \emph{a priori} estimate of $x_k$ with measurement set $\{y_{k-1}\}$. From the perspective of MAP, we have
\begin{equation}\nonumber
\arg \max_{x_k} p(x_{k}|\{y_{k}\}) = \arg \max_{x_k}  p(y_k|x_k)p(x_k|\{y_{k-1}\}).
\end{equation}
Since $w_k$ and $v_k$ are Gaussian (see the assumptions in \eqref{linear}), $p(y_k|x_k)$ and $p(x_k|\{y_{k-1}\}$ should also follow the Gaussian distribution after a linear transformation (see~\cite{b33}) with
\begin{equation}
\begin{small}
\begin{aligned}
p(y_k|x_k)&=\frac{\exp\big(-(y_k-C {x}_{k})^{\prime}R_{k}^{-1}(y_k-C {x}_{k})\big)}{\sqrt{(2\pi)^{m}|R_k|}}\\
p(x_k|y_{k-1})&=\frac{\exp\big(-(x_{k}-A\hat{x}_{k-1})^{\prime}(P_{k}^{-})^{-1}(x_{k}-A\hat{x}_{k-1})\big)}{\sqrt{(2\pi)^{n}|P_{k}^{-}|}}
\label{pdf}
\end{aligned}
\end{small}
\end{equation}
where $\hat{x}_{k-1}$ is the \emph{a posteriori} estimate of the state at time step $k-1$, $|R_k|$ is the determinant  of $R_k$, $P_{k}^{-}$ is the \emph{a priori} estimate of error covariance, and $|P_{k}^{-}|$ is the determinant  of $P_{k}^{-}$. Due to the fact that the normalization constants in the denominator of \eqref{pdf} are independent with the argument ${x}_{k}$, they can be ignored which follows that
\begin{equation}
\begin{small}
\begin{aligned}
&\arg \max_{{x}_{k}} p(x_{k}|\{y_{k}\}) =\arg \max_{{x}_{k}} \exp\big(-(y_k-C {x}_{k})^{\prime}R_{k}^{-1}\\
&\times(y_k-C {x}_{k})\big) \exp\big(-({x}_k-A\hat{x}_{k-1})^{\prime}(P_{k}^{-})^{-1}({x}_k-A\hat{x}_{k-1})\big).
\end{aligned}
\end{small}
\end{equation}
It is equivalent to minimizing the negative log:
\begin{equation}
\begin{aligned}
\arg \min_{{x}_{k}} J_{KF} =&\|R_{k}^{-1/2}(y_k-C {x}_{k})\|_{2}^{2}+\\
&\|(P_{k}^{-})^{-1/2}({x}_{k}-A\hat{x}_{k-1})\|_{2}^{2}.
\label{2norm}
\end{aligned}
\end{equation} 
By defining the measurement error $e_{r,k}$ and process error $e_{p,k}$ as
\begin{equation}
\begin{aligned}
e_{r,k} &\triangleq R_{k}^{-1/2}(y_k-C {x}_{k})\\
e_{p,k} &\triangleq (P_{k}^{-})^{-1/2}({x}_{k}-A_k\hat{x}_{k-1}),
\end{aligned}
\end{equation}
we obtain
\begin{equation}
\begin{aligned}
\arg \min_{{x}_{k}} J_{KF}=&\|e_{r,k}\|_{2}^{2}+\|e_{p,k}\|_{2}^{2}.
\label{2norm1}
\end{aligned}
\end{equation}
\begin{theorem}
	The KF in equations \eqref{kf1}-\eqref{p1} can be derived by the MSE criterion using \eqref{2norm1}.
	\label{Theorem6}
\end{theorem}
The proof of this theorem can be found in some existing works~\cite{b33,b45,d21}. Equations \eqref{iterBaye}-\eqref{2norm1} reveal that KF is optimal for  a linear system with Gaussian noises from the perspective of MAP. However, when the noises $w_k$, $v_k$ are \textbf{non-Gaussian}, the probability density function (PDF) in \eqref{pdf} does not hold. In this case, the $\ell_2$ norm-based loss function is not the best. By analogy the negative logarithm relationship between the noise distribution and the loss function in \eqref{pdf} and \eqref{2norm}, we find that the GL induces the following distribution:
\begin{equation}
\begin{aligned}
p(y_k|x_k)&=\left\{\begin{array}{l}
c_r \exp(-J_{GL,\alpha,\beta_r}(e_{r,k})), e_{r,k} \in \mathscr{Y}\\
0, otherwise
\end{array}\right.\\
p(x_k|y_{k-1})&=\left\{\begin{array}{l}
c_p \exp(-J_{GL,\alpha,\beta_p}(e_{p,k})),  e_{p,k} \in \mathscr{X}\\
0, otherwise
\end{array}\right.
\label{pdf1}
\end{aligned}
\end{equation}
where $J_{GL,\alpha,\beta_r}(\cdot)$ and $J_{GL,\alpha,\beta_p}(\cdot)$ are the generalized loss functions with the $N=1$ (details shown in \eqref{gcloss1}), $\alpha \in \mathbb{R}$, $\beta_r =[\beta_{n+1},\beta_{n+2},\cdots,\beta_{n+m}]^{\prime} \in \mathbb{R}^{m}$,  and $\beta_p =[\beta_1,\beta_2,\cdots,\beta_n]^{\prime} \in \mathbb{R}^{n}$. The symbol $\mathscr{Y}$ is the domain of $e_{r,k}$, $\mathscr{X}$ is the domain of $e_{p,k}$, and $c_r$ and $c_p$ are two constants so that $p(y_k|x_k)$ and $p(x_k|y_{k-1})$ are two proper distributions. The error in \eqref{pdf1} is assumed to be bounded and this assumption is reasonable in practical applications. Compared with the Gaussian distribution in \eqref{pdf}, equation \eqref{pdf1} can represent a wide range of noise distributions. By this assumption with MAP, we have
\begin{equation}\nonumber
\arg \max_{{x}_{k}} p(x_{k}|\{y_{k}\}) = \arg \min_{{x}_{k}} J_{GL,KF}
\end{equation}
with 
\begin{equation}
J_{GL,KF}=J_{GL,\alpha,\beta_r}(e_{r,k})+ J_{GL,\alpha,\beta_p}(e_{p,k}).
\label{GLKF}
\end{equation}One can see that the $\ell_2$-norm based objective function in \eqref{2norm} is replaced by the GL function \eqref{GLKF}. To simplify the visualization of the noise distributions in \eqref{pdf1}, in one dimensional case, we have
\begin{equation}
\begin{aligned}
p(e)&=\left\{\begin{array}{l}
c \exp(-J_{GL,\alpha,\beta}(e)), e \in \mathscr{E}\\
0, otherwise
\end{array}\right..
\label{distribution}
\end{aligned}
\end{equation}
A comparison of the $p(e)$, the Laplace distribution,  the Gaussian distribution, and the $\varepsilon$-contaminated mixture model is 
shown in Fig. \ref{pdfcom}. One can see that $p(e)$ approaches $\mathcal{L}(0,1)$ with $\alpha=1$ and $\beta=100$ in Fig. \ref{lap}, and is close to  $\mathcal{N}(0,0.5)$ with  $\alpha=2$ and $\beta=100$ in Fig. \ref{gau}. Moreover, when selecting a proper bandwidth, it can approach a $\varepsilon$-contaminated mixture model effectively (see the magenta and the dot blue lines in Fig. \ref{lap} and Fig. \ref{gau}). Actually, $p(e)$ approaches a $\alpha$-order exponential distribution when setting $\beta \to \infty$ since $\lim\limits_{\beta \to \infty}c\exp(-J_{GL,\alpha,\beta}(e))=c\exp({-e^{\alpha}})$. When setting a relative small bandwidth, it represents a heavy-tailed distribution with order $\alpha$. This implies that the shape of $p(e)$ can be controlled by the bandwidth flexibly and it is a suitable representation for a large number of distributions.  
\begin{figure}[htbp]
	\centering
	\subfigure[]{
		\begin{minipage}[t]{0.49\linewidth}
			\centering
			\includegraphics[width=1.0\columnwidth]{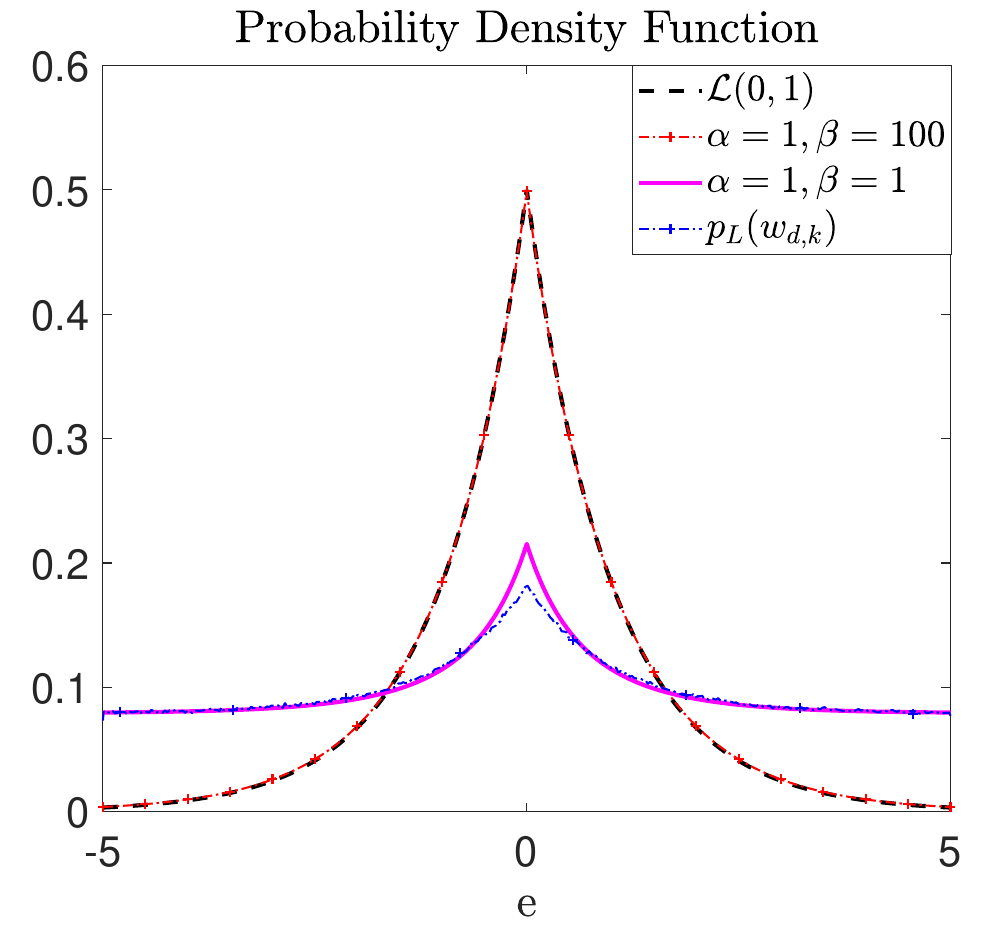}
			\label{lap}
		\end{minipage}%
	}%
	\subfigure[]{
		\begin{minipage}[t]{0.505\linewidth}
			\centering
			\includegraphics[width=1.0\columnwidth]{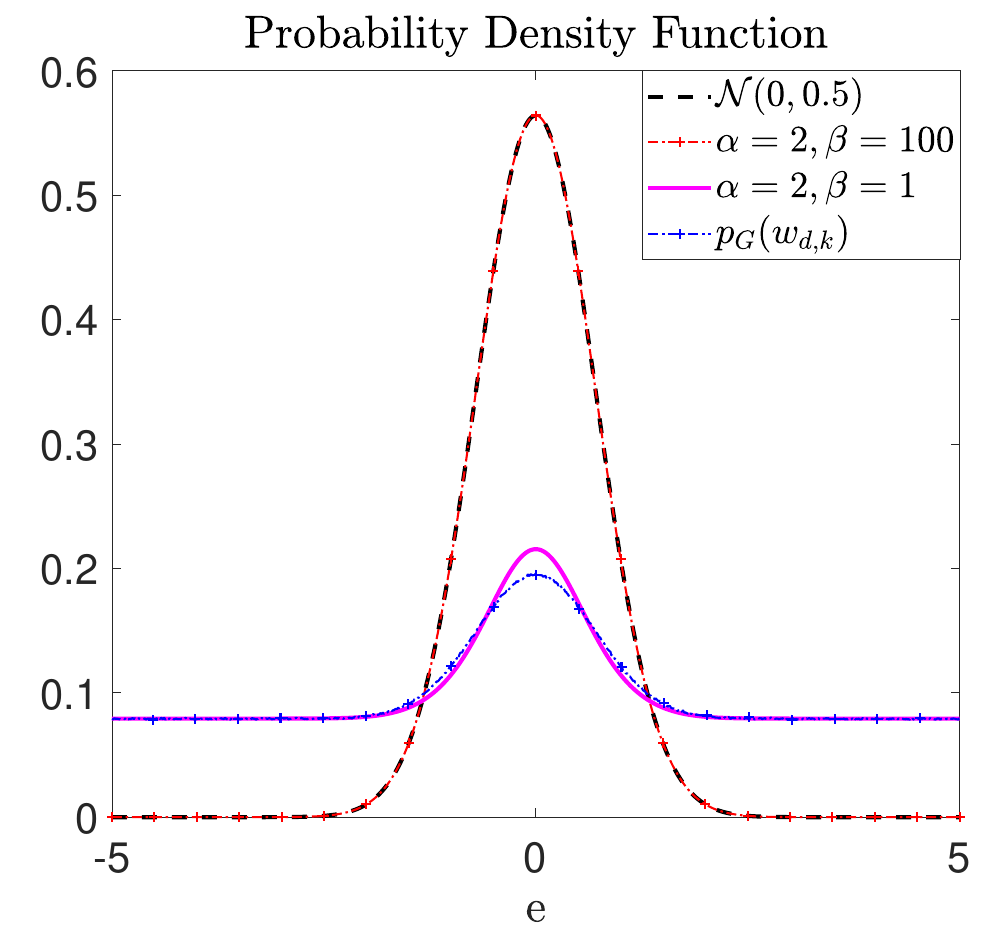}
			\label{gau}
		\end{minipage}%
	}%
	\caption{Laplace distribution, Gaussian distribution, $p(e)$ in \eqref{distribution} with different $\alpha$ and $\beta$, and $p(w_{d,k})$ using a $\varepsilon$-contaminated mixture model. The error domain is set to be $\mathscr{E}=[-5,5]$ in $p(e)$. The noise distribution $p_{L}(w_{d,k})$ follows $w_{d,k}\sim 0.2\mathcal{L}(0,1)+0.8\mathcal{U}(-5,5)$ while $p_{G}(w_{d,k})$ follows $w_{d,k}\sim 0.2\mathcal{N}(0,0.5)+0.8\mathcal{U}(-5,5)$.}	
	\label{pdfcom}
\end{figure}
\begin{remark}
	When the kernel bandwidth $\beta$ is very small, the distribution \eqref{distribution} is very similar to a uniform distribution. However, a very small kernel bandwidth may bring difficulty in the convergence when solving \eqref{GLKF} (see Theorems \ref{theorem8} and \ref{theorem9} in the following section). 
\end{remark}
\begin{remark}
	The GL has two advantages compared with the traditional correntropy-based loss function. Firstly, it has a shape parameter to tune so that it can accommodate many types of distributions. Secondly, it employs different kernel bandwidths at distinct channels. Moreover, the GL is associated with the noise distribution through \eqref{pdf1}, which provides the general guidance for kernel parameter selection, i.e., a smaller kernel bandwidth corresponds to a heavier tail distribution. A more detailed kernel parameter tuning strategy is available in Section \ref{tuning}. 
\end{remark}

\section{Algorithm Derivation}
In this section, we derive the GMKMCKF, analyze its convergence and complexity. Then, we apply this algorithm to disturbance estimation.
\subsection{Algorithm Derivation}
The system dynamics in \eqref{linear} can be rewritten as 
\begin{equation}
\left(\begin{array}{c}
{x}_k^{-}\\
{y}_k
\end{array}\right)=
\left(\begin{array}{c}
{I}\\
{C}
\end{array}\right){x}_k+{\nu}_k
\label{mcc}
\end{equation}
where ${x}_k^{-}$ is the \textit{a priori} estimate of state which can be obtained by \eqref{kf1}. The noise ${\nu}_k$ has
\begin{equation}\nonumber
{\nu}_k=\left(\begin{array}{c}
{x}_k^{-}-{x}_k\\
{v}_k
\end{array}\right)
\end{equation}
with 
\begin{equation}\nonumber
\begin{small}
\begin{aligned}
E({\nu}_k {\nu}_k^{\prime})&=\left(\begin{array}{cc}
{P}_{k}^{-}&0\\
0&{R_k}
\end{array}\right)=\left(\begin{array}{cc}
{B}_{p}{B}_{p}^{\prime}&0\\
0&{B}_{r}{B}_{r}^{\prime}
\end{array}\right)={B}_{k}{B}_{k}^{\prime}
\label{bk}
\end{aligned}
\end{small}
\end{equation}
where ${P}_{k}^{-}$ is the \textit{a priori} error covariance, and ${B}_{p}$ and ${B}_{r}$ can be obtained by Cholesky decomposition. Left multiplying ${B}_{k}^{-1}$ in both sides of (\ref{mcc}), we obtain
\begin{equation}
{T}_k={W}_k{x}_k+{\zeta}_k
\label{TWexpression}
\end{equation}
with
\begin{equation}
\begin{aligned}
{T}_k&={B}_{k}^{-1}\left(\begin{array}{c}
{x}_k^{-}\\
{y}_k
\end{array}\right),
{W}_k={B}_{k}^{-1}\left(\begin{array}{c}
{I}\\
{C}
\label{TW}
\end{array}\right)
\end{aligned}
\end{equation}
and $\zeta_k={B}_{k}^{-1}{\nu}_k$. Using the GL in \eqref{GLKF} as the loss function, we have
\begin{equation}\nonumber
\arg \min_{x_k} J_{GL,KF} = \sum_{i=1}^{n+m} \beta_{i}^{\alpha}\left(1-G_{\alpha,\beta_i}({e}_{i,k})\right)
\label{objgkf0}
\end{equation}
where ${e}_{i,k}={t}_{i,k}-{w}_{i,k}{x}_k$ is the error at time step $k$, ${t}_{i,k}$ is the $i$-th element of ${T}_k$, ${w}_{i,k}$ is the $i$-th row of ${W}_k$, and $\alpha$ and $\beta_i$ are kernel parameters. It follows that
\begin{equation}
\arg \min_{x_k} J_{GL,KF} = \arg \max_{x_k} J_{GC,KF}
\label{objgkf}
\end{equation}
with
\begin{equation}\nonumber
J_{GC,KF} =\sum_{i=1}^{n+m} \beta_{i}^{\alpha}G_{\alpha,\beta_i}({e}_{i,k}).
\end{equation}
Equation \eqref{objgkf} can be solved by 
\begin{equation}\nonumber
\frac{\partial J_{GC,KF}}{\partial {x}_k}=0
\end{equation}
and it follows that
\begin{equation}\nonumber
\sum_{i=1}^{n+m} w_{i,k}^{\prime}(t_{i,k}-w_{i,k}x_{k})\alpha|e_{i,k}|^{\alpha-2}\exp(-\beta_{i}^{-\alpha} |e_{i,k}|^{\alpha})=0.
\end{equation}
We denote $|e_{i,k}|^{\alpha-2}\exp(-\beta_{i}^{-\alpha} |e_{i,k}|^{\alpha})$ as $g_{C}(e_{i,k})$. Then, we have
\begin{equation}\nonumber
\begin{aligned}
\sum_{i=1}^{n+m}w_{i,k}^{\prime}g_{C}(e_{i,k}) t_{i,k}=\sum_{i=1}^{n+m}w_{i,k}^{\prime}g_{C}(e_{i,k})w_{i,k}x_{k}.
\end{aligned}
\end{equation}
It is easy to obtain that 
\begin{equation}
x_{k}=\left(\sum_{i=1}^{n+m}w_{i,k}^{\prime}g_{C}(e_{i,k})w_{i,k}\right)^{-1}\left(\sum_{i=1}^{n+m}w_{i,k}^{\prime}g_{C}(e_{i,k})t_{i,k}\right).
\label{fixp}
\end{equation}
One can see that the above question is a fixed-point equation since 
both sides of \eqref{fixp} contain $x_{k}$ (note that $e_{i,k}=t_{i,k}-w_{i,k}x_{k}$ is a function of $x_{k}$). It can be expressed as 
\begin{equation}
x_{k}=(W_{k}^{\prime}M_{k}W_{k})^{-1}(W_{k}^{\prime}M_{k}T_{k})
\label{plant}
\end{equation}
with
\begin{equation}\nonumber
M_{k}=\left[\begin{array}{cc}
M_{p}&0\\
0&M_{r}
\end{array}\right]
\end{equation}
where $M_{p}=diag(g_{C}(e_{1,k}),\ldots,g_{C}(e_{n,k}))$ and $M_{r}=diag(g_{C}(e_{n+1,k}),\ldots,g_{C}(e_{n+m,k}))$. Substituting the expression of $W_{k}$ from \eqref{TW} into \eqref{plant}, we have
\begin{equation}\nonumber
\begin{aligned}
({W}_k^{\prime}{M}_k{W}_k)^{-1}=&[({B}_{p}^{-1})^{\prime}{M}_{p}{B}_{p}^{-1}+{C}^{\prime}({B}_r^{-1})^{\prime}{M}_{r}{B}_{r}^{-1}{C}]^{-1}.
\end{aligned}
\end{equation}
Using the matrix inversion lemma, we arrive at
\begin{equation}
\begin{aligned}
({W}_k^{\prime}{M}_k{W}_k)^{-1}&={B}_{p}{M}_{p}^{-1}{B}_{p}^{\prime}-{B}_{p}{M}_{p}^{-1}{B}_{p}^{\prime}{C}^{\prime} ({B}_{r}{M}_{r}^{-1}{B}_{r}^{\prime}\\
&+{C}{B}_{p}{M}_{p}^{-1}{B}_{p}^{\prime}{C}^{\prime})^{-1}{C}{B}_{p}{M}_{p}^{-1}{B}_{p}^{\prime}.
\label{part1}
\end{aligned}
\end{equation}
Further, we have
\begin{equation}
\begin{aligned}
{W}_k^{\prime}{M}_k{T}_k=({B}_{p}^{-1})^{\prime}{M}_{p}{B}_{p}^{-1}{x}_{k}^{-}+{C}^{\prime}({B}_{r}^{-1})^{\prime}{M}_{r}{B}_{r}^{-1}{y}_{k}.
\end{aligned}
\label{part2}
\end{equation} 
Substituting the \eqref{part1} and \eqref{part2} into \eqref{plant}, we have
\begin{equation}
\begin{aligned}
{x}_{k}={x}_{k}^{-}+\tilde{{K}}({y}_{k}-{C}{x}_{k}^{-})
\label{mkmckf}
\end{aligned}
\end{equation} 
with
\begin{equation}
\begin{aligned}
\tilde{{K}}&=\tilde{{P}}_{k}^{-}{C}^{\prime}({C}\tilde{{P}}_{k}^{-}{C}^{\prime}+\tilde{{R}}_{k})^{-1}\\
\tilde{{P}}_{k}^{-}&={B}_{p}{M}_{p}^{-1}{B}_{p}^{\prime},~ \tilde{{R}}_{k}={B}_{r}{M}_{r}^{-1}{B}_{r}^{\prime}.
\end{aligned}
\end{equation}
The \emph{a posteriori} error covariance is given as 
\begin{equation}
{P}_{ k}^{+}=({I}-\tilde{{K}}_k {C}){{P}}_{ k}^{-}({I}-\tilde{{K}}_k {C})^{\prime}+\tilde{{K}}_k {R}_{k}\tilde{{K}}_k^{\prime}.
\end{equation} 
The detailed algorithm of the GMKMCKF is summarized in Algorithm~\ref{AlgorithmMK}. 
\begin{algorithm}[bt]
	\setstretch{0.9} 
	\caption{GMKMCKF}
	\label{AlgorithmMK}
	\begin{algorithmic}[1]
		\State \textbf{Step 1: Initialization}\\
		Choose $\alpha$, $\beta_1,\beta_2,\ldots,\beta_{n+m}$, maximum iteration number $m_{iter}$, and a threshold $\varepsilon$.
		\State {\textbf{Step 2: State Prediction}}\\
		$\hat{{x}}_{k}^{-}=A \hat{{x}}_{k-1}^{+} $\\
		${P}_{ k}^{-}={A} {P}_{k-1 }^{+} {A}^{\prime}+{Q}_k$\\
		Obtain ${B}_{p}$ with ${P}_{ k}^{-}={B}_{p}{B}_{p}^{{\prime}}$\\
		Obtain ${B}_{r}$ with ${R}_{k}={B}_{r}{B}_{r}^{{\prime}}$
		\State \textbf{Step 3: State Update}\\
		$\hat{{x}}_{k,0}^{+}=\hat{{x}}_{k}^{-}$
		\While{$\frac{\left\|\hat{{x}}_{k,t}^{+}-\hat{{x}}_{k,t-1}^{+}\right\|}{\left\|\hat{{x}}_{k,t}^{+}\right\|}>\varepsilon$ or $t \le m_{iter}$} \\
		$\hat{{x}}_{k,t}^{+}=\hat{{x}}_{k}^{-}+\tilde{{K}}_{k,t}({y}_k-{H}\hat{{x}}_{k}^{-})$ \Comment{$t$ starts from 1}\\
		$\tilde{{K}}_{k,t}=\tilde{{P}}_{ k}^{-}{H}^{\prime}({H}\tilde{{P}}_{ k}^{-}{H}^{\prime}+\tilde{{R}}_{k})^{-1}$\\
		$\tilde{{P}}_{ k}^{-}={B}_{p}\tilde{{M}}_{p}^{-1}{{B}}_{p}^{\prime}$\\
		$\tilde{{R}}_{ k}={B}_{r}\tilde{{M}}_{r}^{-1}{B}_{r}^{\prime}$\\
		$ {M}_{p}={diag}(g_{C}(e_{1,k}),\ldots,g_{C}(e_{n,k}))$\\
		$ {M}_{r}={diag}(g_{C}(e_{n+1,k}),\ldots,g_{C}(e_{n+m,k}))$\\
		${e}_{i,k}=t_{i,k}-w_{i,k}x_{k,t-1}^{+}$ \\
		$t=t+1$
		\EndWhile\\
		${P}_{k}^{+}=({I}-\tilde{{K}}_k {H}){{P}}_{ k}^{-}({I}-\tilde{{K}}_k {H})^{\prime}+\tilde{{K}}_k {R}_{k}\tilde{{K}}_k^{\prime}$		
	\end{algorithmic}
\end{algorithm}

\begin{theorem}
	The GMKMCKF is identical to the KF when $\alpha=2$ and $\beta_i \to \infty$. It is identical to the traditional MCKF~\cite{b18} when $\alpha=2$ and $\beta_1=\beta_2=\cdots=\beta_{n+m}=\sqrt{2}\sigma$. Moreover, it becomes the MKMCKF~\cite{b20} when $\alpha=2$ and $\beta_i=\sqrt{2}\sigma_i$.
	\label{Theorem7}
\end{theorem}
The proof of this theorem is shown in Appendix \ref{proofTheorem7}
\subsection{Convergence Issue}
The religious convergence of the GMKMCKF remains open. In this section, we provide a sufficient condition under which the fixed-point iteration  \eqref{fixp} surely converges to a unique solution when setting $\alpha=2$ . We drop the subscript $k$  and use $l=n+m$ for ease of notation. Then, \eqref{fixp} can be rewritten as
\begin{equation}
\begin{aligned}
x&=f(x)=R_{ww}^{-1} P_{wt}\\
&=\left(\sum_{i=1}^{l}{w}_{i}^{\prime}g_{C}({e}_{i}){w}_{i}\right)^{-1}\left(\sum_{i=1}^{l}{w}_{i}^{\prime}g_{C}({e}_{i}){t}_{i}\right)\\
&\overset{\alpha=2}{=}\left(\sum_{i=1}^{l}{w}_{i}^{\prime}G_{\beta_i}({e}_{i}){w}_{i}\right)^{-1}\left(\sum_{i=1}^{l}{w}_{i}^{\prime}G_{\beta_i}({e}_{i}){t}_{i}\right)
\label{cmcxkfree}
\end{aligned}
\end{equation}
where $g_{C}({e}_{i})=G_{\beta_i}({e}_{i})=\exp^{-e_i^2/\beta_i^{2}}$ when setting $\alpha=2$, $e_i=t_i-w_i x$, $w_i \in \mathbb{R}^{1 \times n}$, and $\beta_i$ is the kernel bandwidth for $i$-th channel. We assume that $R_{ww}$ is invertible with $\lambda_{\min}[R_{ww}]>0$ for any value of $\beta_i$ where $\lambda_{\min}[\cdot]$ denotes the minimum eigenvalue of a matrix for tractability. Then, we present the following lemma.
\begin{lemma}
	Based on contraction mapping theorem (also known as Banach fixed-point theorem)~\cite{b43}, the convergence of the fixed-point	algorithm \eqref{cmcxkfree} is guaranteed if $\exists~ \gamma>0$ and $0<\eta<1$ such that the initial vector $\|x_{0}\|_p < \gamma$, and $\forall x \in \{x \in \mathbb{R}^{n}: \|x\|_p \le \gamma\}$, it holds that
	\begin{equation}
	\begin{aligned}
	\left\{\begin{array}{l}
	\|f(x)\|_p \le \gamma\\
	\| \nabla_{x} f(x) \|_p \le \eta
	\end{array}\right.
	\end{aligned}
	\end{equation}
	where $\|\cdot \|$ denotes an $\ell_p$ norm of a vector or an induced norm of a matrix defined by $\|A\|_p = \max \limits_{\|x\|_p} \frac{\|Ax\|_p}{\|x\|_p}$ with $p \ge 1$, $A \in \mathbb{R}^{n \times n}$, $x \in \mathbb{R}^{n \times 1}$, and $\nabla_x f(x)$ is the Jacobian matrix of $f(x)$ given by 
	\begin{equation}\nonumber
	\nabla_{x} f(x)= \begin{bmatrix}\frac{\partial }{\partial {x}_1} f(x), \frac{\partial }{\partial {x}_2} f(x), \cdots, \frac{\partial }{\partial {x}_n} f(x)\end{bmatrix}
	\label{jacobian}
	\end{equation}
	with $x=[{x}_1,{x}_2,\cdots,{x}_n]^{\prime}$.
\end{lemma}
Denote bandwidth vector as $\bar{\beta}=[\beta_1, \beta_2, \ldots, \beta_l]^{\prime} \in \mathbb{R}^{l}$ and the unified bandwidth as $\beta_1=\beta_2=\cdots=\beta_l=\beta \in \mathbb{R}$. Then, we have the following two theorems.
\begin{theorem}
	If $\gamma > \xi$, where $\xi = \frac{\sqrt{n} \sum_{i=1}^{l}\|w_i^{\prime}\|_1 |t_i|}{\lambda_{min}[\sum_{i=1}^{l}w_i^{\prime}w_i]}$, and $\beta_i \ge \beta^{*}$ for $i=1,2,\cdots,l$, where $\beta^{*}$ is the solution of equation $\phi(\beta)=\gamma$ with
	\begin{equation}
	\phi(\beta)= \frac{\sqrt{n}\sum_{i=1}^{l}\|w_i^{\prime}\|_1 |t_i|}{\lambda_{min}\Big{[}\sum_{i=1}^{l} w_i^{\prime}G_{\beta}\big(\gamma \|w_i^{\prime}\|_1+t_i\big)w_i\Big{]}},
	\label{phibeta}
	\end{equation} 
	then $\|f(x)\|_1 \le \gamma$ for all $x \in \{x \in \mathbb{R}^{n}: \|x\|_1 \le \gamma\}$.
	\label{theorem8}
\end{theorem}

\begin{proof}
	The proof of this theorem is shown in \ref{proofTheorem8}.
\end{proof}
\begin{theorem}
	If $\gamma > \xi = \frac{\sqrt{n}\sum_{i=1}^{l}{\|}{w}_i^{\prime}{\|_1}|{t}_i|}{\lambda_{min}\Big{[}\sum_{i=1}^{l}{w}_i^{\prime}{w}_i\Big{]}}$, and $\forall i,~ \beta_i \ge \max\{\beta^{*},\beta^{+}\}$, where $\beta^{*}$ is the solution of $\phi(\beta)=\gamma$ [see $\phi(\beta)$ in \eqref{phibeta}], and $\beta^{+}$ is the solution of $\psi(\beta)=\eta ~(0<\eta<1)$ with
	\begin{equation}\nonumber
	\begin{aligned}
	{\psi}({\beta})= \frac{2\sqrt{n} \sum_{i=1}^{l}(|t_i|+\gamma\|w_i^{\prime}\|_{1}) \|w_i^{\prime}\|_1  \big{(}\gamma\|w_i^{\prime}w_i \|_1+ \|w_i^{\prime}t_i\|_1 \big{)} }{\beta^2\lambda_{\min}\Big{[} 
	\sum_{i=1}^{l}{w}_i^{\prime}G_{\beta}\Big{(}\gamma \|w_i^{\prime}\|_1+|t_i|\Big{)}{w}_i\Big{]} },
\end{aligned}
\end{equation}
	then it holds that $\|f(x)\|\le \gamma$, and $\|\nabla_{x} f(x)\| \le \eta$ for all $x \in \{x \in \mathbb{R}^{n}: \|x\|_1 \le \gamma\}$.
	\label{theorem9}
\end{theorem}

\begin{proof}
	The proof of this theorem is in \ref{proofTheorem9}.
\end{proof}
\begin{remark}
	Theorems \ref{theorem8} and \ref{theorem9} are extensions of Theorem 1 and Theorem 2 in \cite{b44}, which give a sufficient condition for the fixed-point iteration of the GMKMCKF with $\alpha=2$. By Theorem \ref{theorem9} and the contraction mapping theorem~\cite{b43}, given the initial condition $\|x_0\|_1 < \gamma$, the fixed-point algorithm \eqref{cmcxkfree} will surely converge to a unique solution provided that $\alpha=2$ and $\beta_i$ is larger than a certain value and the value of $\eta$ guarantees the convergence speed. Theorem \ref{theorem9} also indicates that the algorithm may diverge if the kernel bandwidth $\beta_i$ is too small, although conceptually a small kernel bandwidth may be more effective in rejecting outliers or disturbance since it corresponds to a much heavier PDF (see Fig. \ref{pdfcom} for details). The rigorous convergence discussion of $\alpha \neq 2$ is ignored in this paper. However, in the simulation (as shown in Fig. \ref{ParameterSweep} of the following section), we observe that the convergence of \eqref{cmcxkfree} holds with a large range of $\alpha$.
\end{remark}
\subsection{Algorithm Complexity and Kernel Parameters Selection}
\label{tuning}
The main computational complexity of the GMKMCKF is summarized in Table \ref{complex}. Note that ${M}_p$ and ${M}_r$ are diagonal matrices and their inverse matrices are easy to compute. Assume that the average iteration number for the while loop in Algorithm \ref{AlgorithmMK} is $\bar{t}$. Then, the computational complexity of the GMKMCKF is
\begin{equation}
\begin{aligned}
S_{our}&= [8n^3+4nm^2+2mn^2-n^2-n+O(n^3)+O(m^3)]\\
&+\bar{t}[2m^3+2n^3+4n^2m+6m^2n+4n^2+2m^2+2mn\\
&+6n+6m+O(m^3)].
\end{aligned}
\end{equation}
Similarly, we can obtain the complexity of the KF in equations \eqref{kf1}-\eqref{p1}, which is 
\begin{equation}
S_{kf}=6n^3+6n^2m+4m^2n+mn-n+O(m^3).
\end{equation}
One can see that the complexity of the GMKMCKF is moderately heavier than that of the KF. In general, the fixed-point algorithm can converge very quickly~\cite{b18} which indicates that the computation complexity of the GMKMCKF is mild. 
\begin{table}[]
	\centering
	\caption{The Computation Complexity of Algorithm \ref{AlgorithmMK}.}
	\begin{tabular}{lll}
	\hline
	\hline
	\begin{tabular}[c]{@{}l@{}}Lines or\\ equations\end{tabular} & \begin{tabular}[c]{@{}l@{}}Absolute value,\\ addition/subtraction,\\and multiplication\end{tabular} & \begin{tabular}[c]{@{}l@{}} exponentiation,\\ exponent, division,\\ and Cholesky\\ decomposition\end{tabular} \\
	\hline
	Line 4& $2n^2-n$ & 0\\
	Line 5& $4n^3-n^2$ & 0\\
	Line 6& 0&$O(n^3)$\\ 
	Line 7&0&$O(m^3)$\\ 
	Line 11&  $4nm$  & 0\\   
	Line 12&$4n^2m+4m^2n-3nm$ &   $O(m^3)$ \\ 
	Line 13&$2n^3$&$n$  \\ 
	Line 14&$2m^3$ &$m$ \\                                             
	Line 15&$3n$ &  $3n$\\ 
	Line 16&$3m$ &$3m$\\                                                                                                                           
	Line 17&$2n$                                                               &$0$ \\
	Line 20&\begin{tabular}[c]{@{}l@{}}$4n^3+4n^2m$\\ $-2n^2+2nm^2$\end{tabular}                                           &   $0$\\
	\eqref{TW}&\begin{tabular}[c]{@{}l@{}}$2m^2n+2n^2+2m^2$\\ $-mn-m-n$\end{tabular}                                           &   $0$\\
	\hline
	\hline                                            
	\end{tabular}
	\label{complex}
\end{table}

In Algorithm \ref{AlgorithmMK}, we have to tune a total of $n+m$ bandwidths and a shape parameter $\alpha$. In general, these parameters can be tuned based on the noise PDF as indicated by \eqref{pdf1}. In the application of disturbance estimation, we can select $\beta_i \to \infty$ for channels without disturbance, and use $\beta_j=c_j$ for channels contaminated by disturbance. The shape parameter $\alpha$ can be tuned based on the shape of the nominal noises (i.e., without considering the disturbance). We can select $1 \le \alpha \le 2$ if the nominal noises are heavy-tailed, use $\alpha=2$ if they are Gaussian, and employ $\alpha>2$ if they are light-tailed. An alternative way to tune the kernel parameters is the optimization algorithm, e.g., Bayesian optimization in ~\cite{b46}.
\section{Simulations}
In this section, we employ the GMKMCKF as a disturbance observer for a robotic manipulator tracking problem. Moreover, we compare it with the ESO~\cite{b8}, KF-DOB~\cite{b11}, MCKF~\cite{b18}, and PF~\cite{b34}.
\subsection{System Modeling}
We consider a one-degree of freedom robotic manipulator tracking problem. The target of the robot is to track a predefined angle $\theta_d$ with or without disturbance $d$. The system dynamics of the robotic manipulator can be written as 
\begin{equation}
{I}_{m} \ddot{\theta} + b_m \dot{\theta} + {k_m}\theta+mgl\sin(\theta)  = \tau+ d
\label{robotdyn}
\end{equation}
where ${I}_{m}$ is the inertia, $m$ is the mass, $l$ is the length of the link, $b_m$ is the damping coefficient, $k_m$ is the stiffness coefficient, $\theta$ is the angle, $g$ is the gravity constant, $\tau$ is the motor output, and $d$ is the disturbance caused by unknown friction or the environment. To eliminate the nonlinear term in \eqref{robotdyn}, we use the feedback linearization technique~\cite{b31} by applying the control input $u_{g}=mgl\sin(\theta)$. In this case, the new model becomes
\begin{equation}
{I}_{m} \ddot{\theta} + b_m \dot{\theta} + k_m\theta= \bar{\tau}+ d.
\label{robotdyn1}
\end{equation}
where $\bar{\tau}=\tau-u_{g}$. Then, equation \eqref{robotdyn1} can be rewritten as a discrete state-space form by Euler discretization
\begin{equation}
\begin{aligned}
x_{k+1}&=A x_{k}+Fu_{k}+w_{k}\\
y_{k}&=Cx_{k}+v_{k}
\label{observer}
\end{aligned}
\end{equation}
with
\begin{equation}\nonumber
\begin{aligned}
A=\left[\begin{array}{ccc}
1&0&0\\
\frac{T}{{I}_{m}}& 1-\frac{b_{m}T}{{I}_{m}} &-\frac{k_{m}T}{{I}_{m}}\\
0&T&1\\
\end{array}\right]\\
F=\left[\begin{array}{c}
0\\
\frac{T}{{I}_{m}}\\
0
\end{array}\right], C=\left[\begin{array}{l}
0,0,1
\end{array}\right]
\end{aligned}
\end{equation}
where $u_{k}=\bar{\tau}_{k}={\tau}_{k}-u_{g,k}$, $x_{k}=[d_{k}, \dot{\theta}_{k},\theta_{k}]^{\prime}$ including the disturbance, the angular velocity, and the angle, $T$ is the sampling time, $w_k=[w_{d,k},w_{\dot{\theta},k},w_{\theta,k}]^{\prime}$ is the process noise, and $v_k$ is the measurement noise. 

In simulation, the desired angle follows $\theta_{d,k}=15 \sin (0.4\pi kT)$. As for the controller, we use a feedforward term $u_{ff,k}$ to compensate for the system dynamics, a feedback controller $u_{fb,k}$ to stabilize the plant, and a disturbance compensator $u_{d,k}$ to counteract the disturbance. The feedback controller is the PD controller~\cite{b32}. The overall controller has
\begin{equation}\nonumber
u_{k}=u_{ff,k}+u_{d,k}+u_{fb,k}
\label{controller}
\end{equation}
with
\begin{equation}\nonumber
\begin{aligned}
\left\{\begin{array}{l}
u_{ff,k}={I}_{m}{\ddot{\theta}}_{d,k}+b{\dot{\theta}}_{d,k}+k{\theta}_{d,k}\\
u_{d,k}=-\hat{d}_{k}\\
u_{fb,k}=k_{p}(\theta_{d,k}-\hat{\theta}_{k})+k_d(\dot{\theta}_{d,k}-\hat{\dot{\theta}}_{k})
\end{array}\right.
\end{aligned}
\end{equation}
where ${\ddot{\theta}}_{d,k}$ is the desired angular acceleration, ${\dot{\theta}}_{d,k}$ is the desired angular velocity, ${\theta}_{d,k}$ is the desired angle, $\hat{\dot{\theta}}_{k}$, $\hat{{\theta}}_{k}$, and $\hat{d}_{k}$ are the estimated angular velocity, angle, and disturbance, $k_{p}$ and $k_d$ are controller gains. The overall motor output is $\tau_{k}=u_{k}+u_{g,k}=u_{k}+mgl\sin{\hat{\theta}_{k}}$. Without considering the external disturbance, it actually is a proportional-–derivative (PD) controller with a feedforward term and its stability is proved in ~\cite{b32}.

In simulation, the disturbance is assumed to be step-like and follows
\begin{equation}\nonumber
{d}_{k}=\left\{\begin{array}{l}
50+{w}_{d,k},400\le k \le 600 \\
{w}_{d,k}, otherwise
\end{array}\right..
\end{equation}
 In simulation, the manipulator inertia is ${I}_{m}=0.1$~Nm$\cdot$s$^2/\deg$, the damping coefficient is $b_m=1$~Nm$\cdot$s$/\deg$, the stiffness coefficient is $k_m=0.1$~Nm$/\deg$, and the sampling time is $T=0.01$ s. We compare the performance of controller \eqref{controller} using the GMKMCKF, KF-DOB, MCKF, ESO, and PF as an observer in two situations: 1) the nominal noises are Laplacian; 2) the nominal noises are Gaussian. In those two cases, we use the same measurement covariance, process covariance, and initial error covariance for the KF-DOB, MCKF, and GMKMCKF. The particle number for the PF is $N=1000$ while the resampling method is the systematic resample. To investigate the error performance of different observers, we conduct 100 independent  Monte Carlo runs for each observer.
\subsection{Laplace Distribution with Unknown Disturbance}
For system dynamics in \eqref{observer} with nominal noise as Laplace distribution, we assume that
\begin{equation}
\begin{aligned}
w_{1,k} \sim \mathcal{L}(0,\frac{0.1\sqrt{2}}{2}),w_{2,k} \sim \mathcal{L}(0,\frac{0.01\sqrt{2}}{2})\\
w_{3,k} \sim \mathcal{L}(0,\frac{0.01\sqrt{2}}{2}), v_{k}\sim \mathcal{L}(0,\frac{0.01\sqrt{2}}{2}).
\label{lapnoise}
\end{aligned}
\end{equation}
It is worth mentioning that the disturbance process noise $w_{1,k}$ in \eqref{lapnoise} is the nominal noise rather than the practical noise since the modeling of the disturbance in \eqref{observer} is not accurate. To this end, we select $\beta_1=1$ to suppress the heavy-tailed noises for the disturbance channel. As for other channels, we use $\beta_2=\beta_3=\beta_4=10^{8}$. We employ the shape parameter $\alpha=1.6$ for the GMKMCKF1 and $\alpha=2$ for the GMKMCKF2. The maximum iteration number in each sample interval is set to be $m_{iter}=5$. The disturbance error using different observers in one Monte Carlo run is shown in Fig. \ref{lap_derror}. The corresponding tracking angle error is shown in Fig. \ref{lap_terror}. The root mean squared errors (RMSE) of the $x_1$ (disturbance), $x_2$ (angular velocity), $x_3$ (angle), $\theta_d-\theta_a$ (tracking error), and the average time consumption of different algorithms are summarized in Table \ref{lapDisturbnace}. These algorithms are executed on MATLAB 2019b on a laptop (Intel i7-8750H, 2.20GHz). One can see that the GMKMCKF1 outperforms the others and has a moderate complexity compared with the other algorithms, which reveals that $\alpha<2$ is suitable for Laplace nominal noises.
\begin{figure}[htbp]
	\centerline{\includegraphics[width=6.0cm]{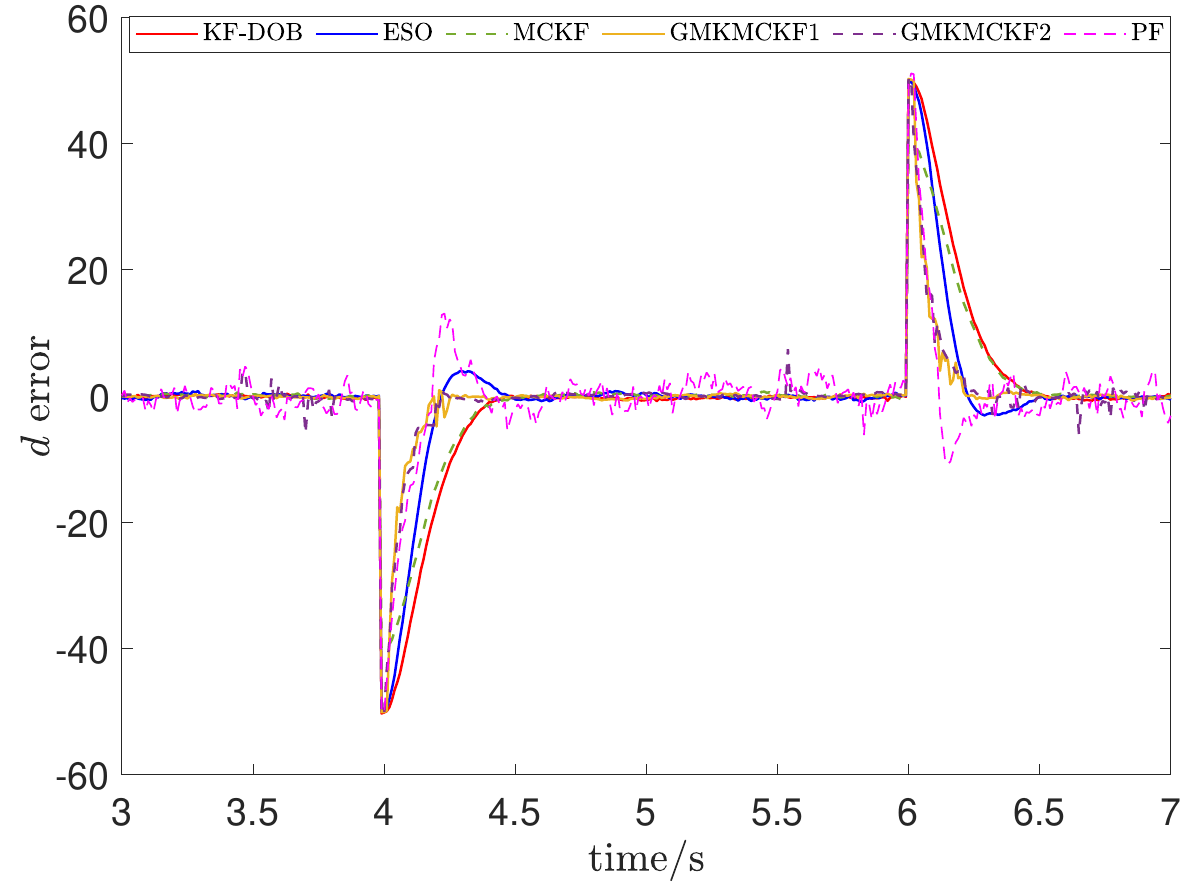}}
	\caption{Disturbance estimation error of different observers. The step-like disturbance is added at $t=4$ seconds and disappears at $t=6$ seconds.}	
	\label{lap_derror}
\end{figure}
\begin{figure}[htbp]
	\centerline{\includegraphics[width=6.0cm]{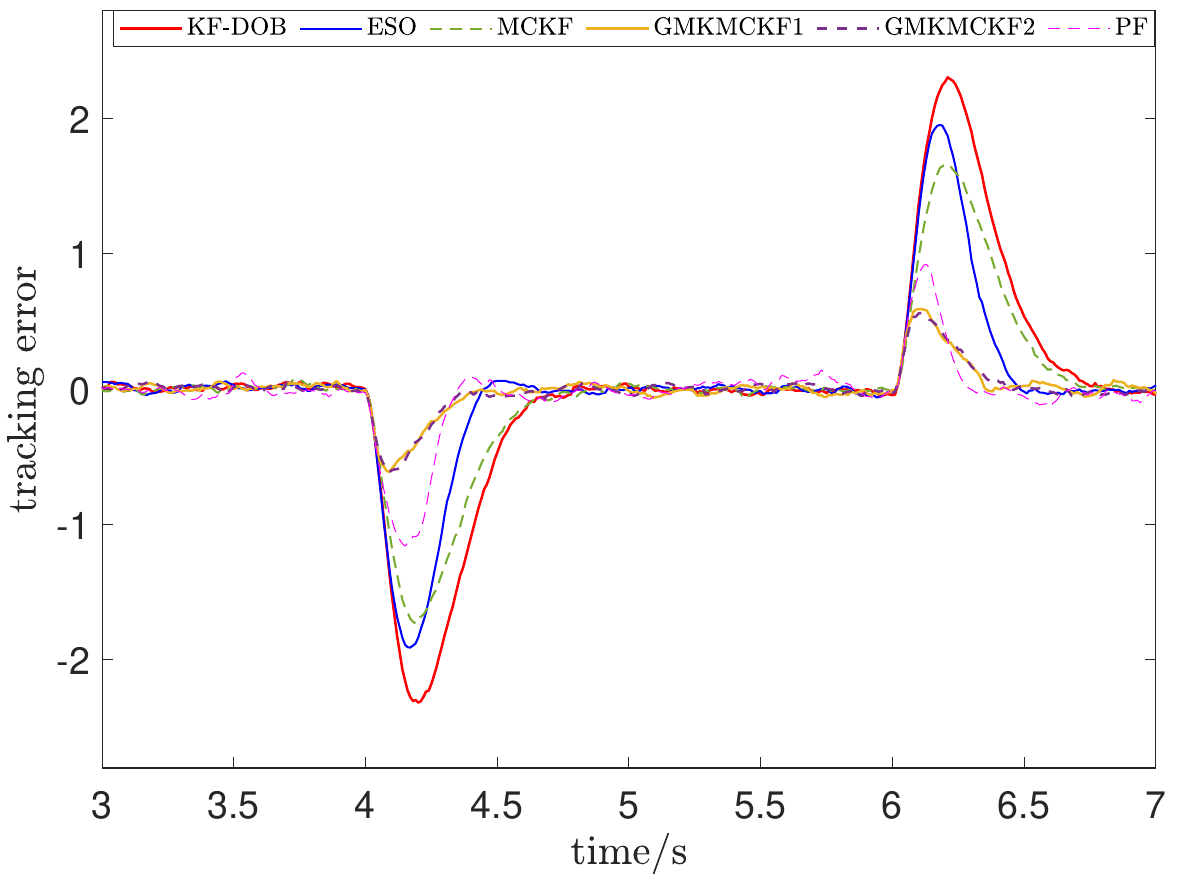}}
	\caption{Tracking error using different observers.}	
	\label{lap_terror}
\end{figure}

\begin{table}[htbp]
	\centering
	\caption{Performance of Different Observers with Laplace Noises.}
	\renewcommand{\arraystretch}{0.95}
	\begin{center}
		\scalebox{0.83}{
		\begin{tabular}{cccccc}
			\hline
			\hline
		\multirow{3}{*}{Observer} & \multirow{3}{*}{\begin{tabular}[c]{@{}c@{}}RMSE\\ of $x_1$\\ (Nm) \end{tabular}} & \multirow{3}{*}{\begin{tabular}[c]{@{}c@{}}RMSE\\ of ${x}_{2}$\\ ($\deg$/s)\end{tabular}} & \multirow{3}{*}{\begin{tabular}[c]{@{}c@{}}RMSE\\ of ${x}_{3}$\\ ($\deg$)\end{tabular}} & \multirow{3}{*}{\begin{tabular}[c]{@{}c@{}}RMSE\\ of $(\theta_d-\theta_a)$\\ ($\deg$)\end{tabular}} & \multirow{3}{*}{\begin{tabular}[c]{@{}c@{}}time\\cost\\ (s)\\ \end{tabular}} \\	                                                                 & & & & &\\
		& & & & &\\
		\hline
			{KF-DOB}&8.0460 & 4.1637&0.0241&0.4951&0.0934\\			
			\hline
			{ESO}& 6.8374 & 3.3209&0.0198&0.3606&0.0927\\
			\hline
			{MCKF}& 6.5352 & 3.4596&0.0249&0.3648&0.1239\\
			\hline
			\textbf{GMKMCKF1}& 	\textbf{4.9361} & 	\textbf{0.7900}&\textbf{0.0083}&\textbf{0.1085}&\textbf{0.1237}\\
			\hline
			{GMKMCKF2}& 	{5.0331} & 	{0.9480}&{0.0088}&{0.1086}&{0.1283}\\
			\hline
			{PF}& 5.3955 & 1.5472&0.0100&0.1473&{3.4388}\\
			\hline
			\hline
		\end{tabular}}
		\label{lapDisturbnace}
	\end{center}
\end{table}
\subsection{Gaussian Distribution with Unknown Disturbance}
We consider the nominal noise as Gaussian distribution for \eqref{observer} with
\begin{equation}
\begin{aligned}
w_{1,k} \sim \mathcal{N}(0,0.01),w_{2,k} \sim \mathcal{N}(0,0.0001)\\
w_{3,k} \sim \mathcal{N}(0,0.0001), v_{k}\sim \mathcal{N}(0,0.0001).
\end{aligned}
\end{equation}
Similarly, we apply $\beta_1=1$ for the disturbance channel, and $\beta_2=\beta_3=\beta_4=10^{8}$ for other channels in the GMKMCKF. Moreover, We employ the shape parameter $\alpha=1.6$ for the GMKMCKF1 and $\alpha=2$ for the GMKMCKF2. The maximum iteration number is set to be $m_{iter}=3$. The RMSE and the average time consumption of different observers are summarized in Table \ref{gauDisturbnace}. One can see the GMKMCKF2 outperforms the others which indicates that $\alpha=2$ is an better option for Gaussian nominal noises.
\begin{table}[htbp]
	\caption{Performance of Different Observers with Gaussian Noises.}
	\renewcommand{\arraystretch}{0.95}
	\begin{center}
	    \scalebox{0.83}{
		\begin{tabular}{cccccc}
		\hline
		\hline
		\multirow{3}{*}{\begin{tabular}[c]{@{}c@{}}Observer\\ \end{tabular}} & \multirow{3}{*}{\begin{tabular}[c]{@{}c@{}}RMSE \\ of ${x}_{1}$\\ (Nm)\end{tabular}} & \multirow{3}{*}{\begin{tabular}[c]{@{}c@{}}RMSE \\ of ${x}_{2}$\\ ($\deg$/s)\end{tabular}} & \multirow{3}{*}{\begin{tabular}[c]{@{}c@{}}RMSE \\ of ${x}_{3}$\\ ($\deg$)\end{tabular}} & \multirow{3}{*}{\begin{tabular}[c]{@{}c@{}}RMSE \\ of $(\theta_d-\theta_a)$\\ ($\deg$)\end{tabular}} & \multirow{3}{*}{\begin{tabular}[c]{@{}c@{}} time\\cost\\(s)\end{tabular}}\\	                             & & & & &\\
		& & & & &\\
			\hline
			{KF-DOB}&8.0397 & 4.1571&0.0241&0.4350&0.007\\			
			\hline
			{ESO}& 6.8507 & 3.3302&0.0198&0.3182&0.006\\
			\hline
			{MCKF}& 6.5286 & 3.4494&0.0250&0.3626&0.0287\\
			\hline
			{GMKMCKF1}& 	4.9871 & 0.8199&0.0083&0.0837&0.0288\\
			\hline
			\textbf{GMKMCKF2}& 	\textbf{4.9193} & 	\textbf{0.7494}&\textbf{0.0082}&\textbf{0.0777}&\textbf{0.0305}\\
			\hline
			{PF}& 5.2685 & 1.3566&0.0097&0.1005&{3.4380}\\
			\hline
			\hline
		\end{tabular}}
		\label{gauDisturbnace}
	\end{center}
\end{table}

To investigate the parameter sensitiveness of the GMKMCKF, we conduct simulations using different $\alpha$ and $\beta_1$ when the nominal noise is Gaussian. The result is shown in Fig. \ref{ParameterSweep}. One can see that the performance of the GMKMCKF is significantly better than the KF-DOB [i.e., RMSE of $x_1$ is 8.0397 as shown in Table \ref{gauDisturbnace}] under a range of kernel parameters. Moreover, we observe that a smaller bandwidth $\beta_1$ is more effective in terms of disturbance mitigation. However, a very small $\beta_1$ may induce the divergence of the fixed-point algorithm (see Theorem \ref{theorem9}). We also find that the estimation result is less sensitive to the shape parameter $\alpha$ compared with $\beta_1$, especially when the bandwidth $\beta_1$ is relatively small. The numerical results indicate that the estimation accuracy is preferable with $\alpha \approx 2$ when the nominal noise is Gaussian.
\begin{figure}[htbp]
	\centerline{\includegraphics[width=7cm]{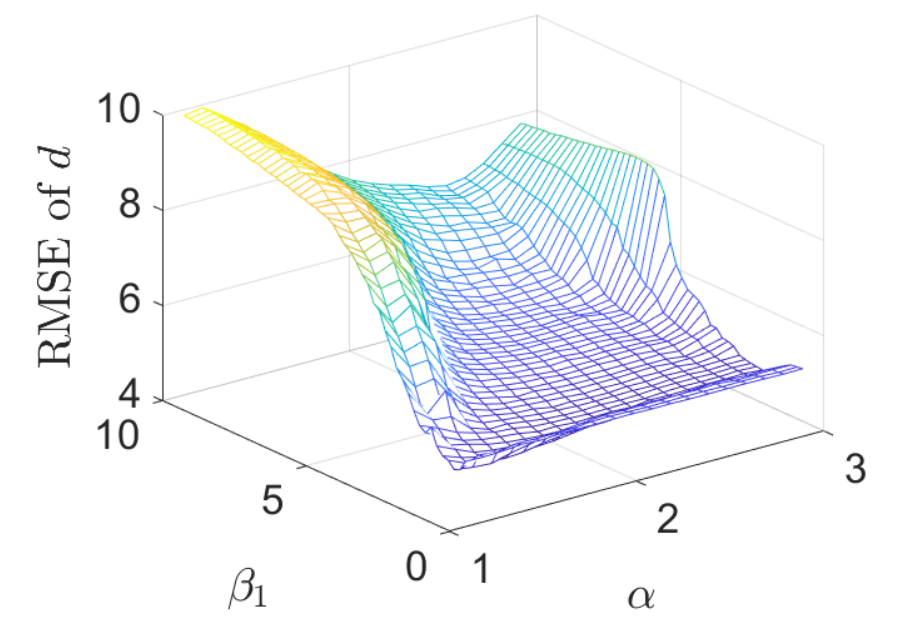}}
	\caption{RMSE of the disturbance (i.e., $x_1$) with different $\alpha$ and $\beta_1$.}	
	\label{ParameterSweep}
\end{figure}

\section{Conclusion}
In this paper, we derive a novel algorithm called the generalized multi-kernel maximum correntropy Kalman filter (GMKMCKF). Our algorithm is derived based on the generalized loss (GL) rather than the conventional mean square error criterion and is capable of situations with some channels contaminated by heavy-tail noise. The proposed algorithm is an extension of the MCKF and MKMCKF but is much more versatile. The convergence of the proposed algorithm can be guaranteed when the kernel bandwidth is bigger than a certain level and its complexity is moderate. Simulations on a robotic manipulator verify the effectiveness of the proposed method. One limitation of this work is that the selection of the kernel parameters is demanding. In the future, we would design adaptive kernel parameter strategies.
\section{Appendix}
\subsection{Proof of Theorem \ref{Theorem1}}
\label{proofTheorem1}
\begin{proof}
	In the case of $0<\alpha \le 2$, the GGD $\kappa_{\alpha,\beta_i}(x_i,y_i)=G_{\alpha,\beta_i}(x_i,y_i)$ is a positive definite kernel (see \cite{b26}, p434), which induces a mapping function $\Phi$ from input space to infinite dimensional reproducing kernel Hibert space (RKHS) with $\kappa_{\alpha,\beta_i}(x_i,y_i)=\Phi(x_{i})^{T}\Phi(y_{i})$. For random pairs $(\mathcal{X}_{i}$, $\mathcal{Y}_{i})$ in \eqref{correntropypair}, we obtain ${C}_{\alpha,\beta_i}(\mathcal{X}_{i},\mathcal{Y}_{i})=\frac{1}{N}\sum_{k=1}^{N}G_{\alpha,\beta_i}\left(x_{i}(k),y_{i}(k)\right)=E\left[\Phi(x_{i})^{T}\Phi(y_{i})\right]$. Thus, $\hat{{C}}(\mathcal{X},\mathcal{Y})=\sum_{i=1}^{l}\beta_i^{\alpha}{C}_{\alpha,\beta_i}(\mathcal{X}_{i},\mathcal{Y}_{i})=\sum_{i=1}^{l}\beta_i^{\alpha}E\left[\Phi(x_{i})^{T}\Phi(y_{i})\right]$. This completes the proof.
\end{proof}
\subsection{Proof of Theorem \ref{Theorem2}}
\label{proofTheorem2}
\begin{proof}
	Equation \eqref{gcloss} can be rewritten as
	\begin{equation}
	\begin{aligned}
	J_{GL}(\mathcal{X},\mathcal{Y})&=\sum_{i=1}^{l}\beta_{i}^{\alpha}\left(1-E\left[G_{\alpha,\beta_i}(x_{i},y_{i})\right]\right)
	\label{GLOSS}
	\end{aligned}
	\end{equation} 
	Taking Taylor series expansion of $G_{\alpha,\beta_i}(x_{i},y_{i})$ and substituting the result into \eqref{GLOSS}, we have
	\begin{equation}\nonumber
	\begin{aligned}
	J_{GL}(\mathcal{X},\mathcal{Y})=\sum_{i=1}^{l}\beta_{i}^{\alpha}\left(1-E\left[\sum_{n=0}^{\infty}\frac{(-1)^{n}}{{\beta}_{i}^{\alpha n}n!}|{x}_{i}-{y}_{i}|^{\alpha n}\right]\right).
	\end{aligned}
	\end{equation}
	It follows that
	\begin{equation}\nonumber
	\lim_{\beta_{i}^{\alpha} \to \infty} J_{GL}(\mathcal{X},\mathcal{Y}) = \sum_{i=1}^{l}E[|{x}_{i}-{y}_{i}|^{\alpha}]=E[\|\mathcal{X}-\mathcal{Y}\|_{\alpha}^{\alpha}]
	\end{equation}
	This completes the proof.
\end{proof}
\subsection{Proof of Theorem \ref{Theorem3}}
\label{proofTheorem3}
\begin{proof}
	Based on the definition of $\mathrm{GCIM}$, we have
	\begin{equation}\nonumber
	\begin{aligned}
	&\mathrm{GCIM}(\mathcal{X},\mathcal{Y})=\bigg(\sum_{i=1}^{l}\beta_{i}^{\alpha}\big(1-{C_{\alpha,\beta_i}}(\mathcal{X}_{i},\mathcal{Y}_{i})\big)\bigg)^{1/2}\\
	&=\bigg(\sum_{i=1}^{l}\beta_i^{\alpha}\Big(1-\frac{1}{N}\sum_{k=1}^{N}G_{\alpha,\beta_i}\big(x_i(k),y_i(k)\big)\Big)\bigg)^{1/2}.
	\end{aligned}
	\end{equation}
	 When $0<\alpha \le 2$, the kernel $\kappa_{\alpha,\beta_i}\big(x_i(k),y_i(k)\big)=G_{\alpha,\beta_i}\big(x_i(k),y_i(k)\big)=\exp \left(-\left|\frac{e_i(k)}{\beta}\right|^{\alpha}\right) \le 1$ is a Mercer kernel (see \cite{b26}, p434) with $e_i(k)=x_i(k)-y_i(k)$, which induces a mapping function $\Phi$ from input space to infinite dimensional reproducing kernel Hilbert space with $\kappa_{\alpha,\beta_i}\big(x_i,y_i\big)=\left<\Phi(x_{i}),\Phi(y_i)\right>_{\mathcal{F}}$. Then, it is clear that $\mathrm{GCIM}$ satisfies: 1) Nonnegativity: $\mathrm{GCIM}(\mathcal{X},\mathcal{Y})\ge 0$; 2) Identities of indiscernibles: $\mathrm{GCIM}(\mathcal{X},\mathcal{Y})=0$ if and only if $\mathcal{X}=\mathcal{Y}$; 3) Symmetry: $\mathrm{GCIM}(\mathcal{X},\mathcal{Y})=\mathrm{GCIM}(\mathcal{Y},\mathcal{X})$. For the triangle inequity: $\mathrm{GCIM}(\mathcal{X},\mathcal{Z}) \le \mathrm{GCIM}(\mathcal{X},\mathcal{Y}) +\mathrm{GCIM}(\mathcal{Y},\mathcal{Z})$. We construct vectors $\tilde{\mathcal{X}}_{i}=[\Phi(x_{i}(1)),\Phi(x_{i}(2)),\ldots,\Phi(x_{i}(N))]^{\prime}$ and $\tilde{\mathcal{Y}}_{i}=[\Phi(y_{i}(1)),\Phi(y_{i}(2)),\ldots,\Phi(y_{i}(N))]^{\prime}$ in Hilbert space $\mathcal{F}^{N}$ for random pairs $(\mathcal{X}_{i}$, $\mathcal{Y}_{i})$. Then, the square of the Euclidean distance $D(\tilde{\mathcal{X}}_i,\tilde{\mathcal{Y}}_i)$ has
	\begin{equation}\nonumber
	\begin{small}
	\begin{aligned}
	&D^{2}(\tilde{\mathcal{X}}_{i},\tilde{\mathcal{Y}}_{i})=\left<\tilde{\mathcal{X}}_{i}-\tilde{\mathcal{Y}}_{i},\tilde{\mathcal{X}}_{i}-\tilde{\mathcal{Y}}_{i}\right>\\
	&=\left<\tilde{\mathcal{X}}_{i},\tilde{\mathcal{X}}_{i}\right>-2\left<\tilde{\mathcal{X}}_{i},\tilde{\mathcal{Y}}_{i}\right>+\left<\tilde{\mathcal{Y}}_{i},\tilde{\mathcal{Y}}_{i}\right>\\
	&=\sum_{k=1}^{N}\kappa_{\alpha,\beta_i}(0)-2\sum_{k=1}^{N}\kappa_{\alpha,\beta_i}(x_{i}(k),y_{i}(k))+\sum_{k=1}^{N}\kappa_{\alpha,\beta_i}(0)\\
	&=2N\left(1-{C}_{\alpha,\beta_i}(\mathcal{X}_{i},\mathcal{Y}_{i})\right)
	\label{rkhs}
	\end{aligned}
	\end{small}
	\end{equation}
	Then, based on the property of Euclidean distance $D(\tilde{\mathcal{X}}_{i},\tilde{\mathcal{Z}}_{i}) \le D(\tilde{\mathcal{X}}_{i},\tilde{\mathcal{Y}}_{i})+D(\tilde{\mathcal{Y}}_{i},\tilde{\mathcal{Z}}_{i})$, and using the Minkowski inequality~\cite{b27}, we have
	\begin{equation}
	\begin{small}
	\begin{aligned}
	\left(\sum_{i=1}^{l}D^{2}(\tilde{\mathcal{X}}_{i},\tilde{\mathcal{Z}}_{i})\right)^{1/2} \le \left(\sum_{i=1}^{l}\left(D(\tilde{\mathcal{X}}_{i},\tilde{\mathcal{Y}}_{i})+D(\tilde{\mathcal{Y}}_{i},\tilde{\mathcal{Z}}_{i})\right)^{2}\right)^{1/2}\\
	\le \left(\sum_{i=1}^{l}D^{2}(\tilde{\mathcal{X}}_{i},\tilde{\mathcal{Y}}_{i})\right)^{1/2}+\left(\sum_{i=1}^{l}D^{2}(\tilde{\mathcal{Y}}_{i},\tilde{\mathcal{Z}}_{i})\right)^{1/2}
	\end{aligned}
	\end{small}
	\end{equation}
	Substituting \eqref{rkhs} into $\mathrm{GCIM}(\mathcal{X},\mathcal{Z})$, we have
	\begin{equation}\nonumber
	\begin{small}
	\begin{aligned}
	&\mathrm{GCIM}(\mathcal{X},\mathcal{Z})=\left(\sum_{i=1}^{l}\beta_{i}^{\alpha}\left(\frac{D^{2}(\tilde{\mathcal{X}}_{i},\tilde{\mathcal{Z}}_{i})}{2N}\right)\right)^{1/2} \\
	&\le\left(\sum_{i=1}^{l}\beta_{i}^{\alpha}\left(\frac{D^{2}(\tilde{\mathcal{X}}_{i},\tilde{\mathcal{Y}}_{i})}{2N}\right)\right)^{1/2} + \left(\sum_{i=1}^{l}\beta_{i}^{\alpha}\left(\frac{D^{2}(\tilde{\mathcal{Y}}_{i},\tilde{\mathcal{Z}}_{i})}{2N}\right)\right)^{1/2}\\
	&=\mathrm{GCIM}(\mathcal{X},\mathcal{Y})+\mathrm{GCIM}(\mathcal{Y},\mathcal{Z}).
	\end{aligned}
	\end{small}
	\end{equation}
	This completes the proof.
\end{proof}

\subsection{Proof of Theorem \ref{Theorem4}}
\label{proofTheorem4}
\begin{proof}
	Taking Taylor series expansion of $G_{\alpha,\beta_i}\left(e_{i}\right)$, one has
	\begin{equation}
	G_{\alpha,\beta_i}\left(e_{i}\right)=\sum_{n=0}^{\infty}\frac{(-1)^{n}}{{\beta}_{i}^{\alpha n}n!}|e_i|^{\alpha n}.
	\label{Taylor}
	\end{equation}
	Substituting \eqref{Taylor} into \eqref{gcloss1}, we have 
	\begin{equation}\nonumber
	\begin{aligned}
	\lim_{\beta_{i}^{\alpha} \to \infty}J_{GL}(e)=\sum_{i=1}^{l}|e_{i}|^{\alpha}=\|e\|_{\alpha}^{\alpha}.
	\end{aligned}
	\end{equation}
	In this case, $J_{GL}(e)$ is identical to $J_{LMP}(e)$ with $\alpha=p$. The Hessian matrix of $J_{GL}(e)$ has 
	\begin{equation}\nonumber
	\begin{aligned}
	H(J_{GL})&=\frac{\partial \nabla J_{GL}}{\partial e}=\begin{bmatrix}
	&\zeta_{1}, &0,&\ldots,0\\
	&0, &\zeta_{1},&\ldots,0\\
	&\vdots, &\vdots,&\ldots,\vdots\\
	&0, &0,&\ldots,\zeta_{l}\\
	\end{bmatrix}
	\end{aligned}
	\end{equation}
	with 
	\begin{equation}\nonumber
	\zeta_{i}=-\frac{\alpha e^{-\frac{|e_{i}|^{\alpha}}{\beta_{i}^{\alpha}}}|e_{i}|^{\alpha}\left(\alpha|e_{i}|^{\alpha}-(\alpha-1) \beta_{i}^{\alpha}\right)}{\beta_{i}^{\alpha} e_{i}^{2}}, i=0,1,\ldots,l.
	\end{equation}
	One can see that $H(J_{GL})$ is a diagonal matrix. When $0<\alpha \le 1$, we have $H(J_{GL}) \prec 0$ for any $e \neq 0$. Thus, $H(J_{GL})$ is concave in this case. When $\alpha>1$ and $|e_{i}|\le (\frac{\alpha-1}{\alpha})^{\frac{1}{\alpha}}\beta_{i}$, we have $H(J_{GL}) \succcurlyeq 0$. Then, $H(J_{GL})$ is convex in this situation. This completes the proof.
\end{proof}

\subsection{Proof of Theorem \ref{Theorem5}}
\label{proofTheorem5}
\begin{proof}
	A differentiable function $f: \mathbb{R}^{l} \rightarrow \mathbb{R}$ is said to be invex, if and only if~\cite{b48}
	\begin{equation}\nonumber
	f(x_2) \ge f(x_1) + q(x_1,x_2)^{T}\nabla f(x_1)
	\end{equation}
	where $\nabla f(x_1)$ is the gradient of $f(x_1)$ with respect to $x_1$, and $ q(x_1,x_2)$ is the vector valued function. In the case of $N=1$, the GL and its gradient are shown in \eqref{gcloss1} and \eqref{infl}. Then, we have $J_{GL}(e)>0$ for any $e \neq \textbf{0}$ and $ J_{GL}(\textbf{0})=0$ where $\textbf{0}$ is the zero vector. This indicates that $J_{GL}(\textbf{0})$ is a global minima of $J_{GL}$. We construct the vector valued function $q (e_1, e_2)$ as follow:
	\begin{equation}\nonumber
	\begin{aligned}
	q (e_1, e_2)=\left\{\begin{array}{cc}
	\frac{J_{GL}(e_2)-J_{GL}(e_1)}{\nabla J_{GL}(e_1)^{T}\nabla J_{GL}(e_1)}\nabla J_{GL}(e_1), ~\quad &e_1 \neq \textbf{0}\\
	\textbf{0}, ~\quad &e_1 = \textbf{0}
	\end{array}\right..
	\end{aligned}
	\end{equation}
	Then, it holds that
	\begin{equation}\nonumber
	J_{GL}(e_2) \ge J_{GL}(e_1) +q (e_1, e_2)^{T}\nabla J_{GL}(e_1)
	\end{equation}
	for any $e_1, e_2 \in \mathbb{R}^{l}$. This completes the proof.
\end{proof}
\subsection{Proof of Theorem \ref{Theorem7}}
\label{proofTheorem7}
\begin{proof}
	When $\alpha=2$, and $\beta_{i} \rightarrow \infty$, we have ${M}_{p}={I}_{n \times n}$ and ${M}_{r}={I}_{m \times m}$. Then, the GMKMCKF is equal to the KF. As $\alpha=2$ and $\beta_{i}=\sqrt{2}\sigma$ for all $i$, we have ${M}_{p}={diag}[G_{\sigma}(e_{1,k},\ldots,e_{n,k})]$, ${M}_{r}={diag}[G_{\sigma}(e_{n+1,k},\ldots,e_{n+m,k})]$. Then, it is identical to the MCKF. When $\alpha=2$ and $\beta_i=\sqrt{2}\sigma_i$, we have ${M}_{p}={diag}[G_{\sigma_{p}}(e_{1,k},\ldots,e_{n,k})]$, ${M}_{r}={diag}[G_{\sigma_{r}}(e_{n+1,k},\ldots,e_{n+m,k})]$ with $\sigma_{p}=[\sigma_1,\ldots,\sigma_n]^{\prime}$ and $\sigma_{r}=[\sigma_{n+1},\ldots,\sigma_{n+m}]^{\prime}$. In this case, it becomes the MKMCKF.
\end{proof}
\subsection{Proof of Theorem \ref{theorem8}}
\label{proofTheorem8}
\begin{proof}
	Due to the fact that the induced norm is compatible with the vector $\ell_p$ norm, then, we have
	\begin{equation}
	\|f(x)\|_1= \|R_{ww}^{-1}P_{wt}\|_1 \le  \|R_{ww}^{-1}\|_1 \|P_{wt}\|_1
	\label{fx}
	\end{equation}
	where   $\|R_{ww}^{-1}\|_1$ is the maximum absolute column of matrix $R_{ww}^{-1}$. Based on the matrix theory, the following inequity holds:
	\begin{equation}
	\|R_{ww}^{-1}\|_1 \le \sqrt{n} 	\|R_{ww}^{-1}\|_2 = \sqrt{n}\lambda_{max}[R_{ww}^{-1}]
	\label{RWW}
	\end{equation}
	where $\|R_{ww}^{-1}\|_2$ is the 2-norm of $R_{ww}^{-1}$ which is equal to the maximum  eigenvalue of the matrix. Then, we have
	\begin{equation}
	\begin{aligned}
	\lambda_{max}[R_{ww}^{-1}]&=\frac{1}{\lambda_{min}[R_{ww}]}=\frac{1}{\lambda_{min}[\sum_{i=1}^{l}{w_i}^{\prime}G_{\beta_i}(e_i){w_i}]}\\ &\overset{(\mathrm{I})}{\le}\frac{1}{\lambda_{min}\Big{[}\sum_{i=1}^{l}{w}_i^{\prime}G_{\beta_i}\Big(|t_i|+\gamma \|w_i^{\prime}\|_1\Big){w}_i\Big{]}}
	\label{RWW_lam}
	\end{aligned}
	\end{equation}
	where (I) comes from $|e_i|=|t_i-w_ix|\le |t_i| + |w_ix| \le |t_i| + \|x\|_1\|w_i^{\prime}\|_1  \le |t_i|+ \gamma \|w_i^{\prime}\|_1$. In addition, it holds that 
	\begin{equation}
	\begin{aligned}
	\|P_{wt}\|_1&=\Big{\|}\sum_{i=1}^{l}{w}_i^{\prime}G_{\beta_i}(e_i){t}_i\Big{\|_1}  \\
	&\overset{(\mathrm{II})}{\le} \sum_{i=1}^{l}\Big{\|}{w}_i^{\prime}G_{\beta_i}(e_i){t}_i\Big{\|_1}\overset{(\mathrm{III})}{\le} \sum_{i=1}^{l}{\|}{w}_i^{\prime}{\|_1}|{t}_i|
	\label{PWT}
	\end{aligned}
	\end{equation}
	where (II) comes from the convexity of $\ell_1$ norm, and (III) comes from $G_{\beta_i}(e_i) \le 1$ for any $e_i$. Substituting \eqref{RWW}, \eqref{RWW_lam}, and \eqref{PWT} into \eqref{fx}, we obtain
	\begin{equation}
	\|f(x)\|_1 \le \bar{\phi}(\bar{\beta}) = \frac{\sqrt{n}\sum_{i=1}^{l}{\|}{w}_i^{\prime}{\|_1}|{t}_i|}{\lambda_{min}\Big{[}\sum_{i=1}^{l}{w}_i^{\prime}G_{\beta_i}\Big{(}\gamma \|w_i^{\prime}\|_1+|t_i|\Big{)}{w}_i\Big{]}}
	\label{phi}
	\end{equation}
	with $\bar{\beta}=[\beta_1,\beta_2,\cdots,\beta_l]^{\prime}$. If we restrain $\beta_1=\beta_2=\cdots=\beta_l=\beta$, equation \eqref{phi} degenerates to a function of $\beta$, i.e.,
	\begin{equation}\nonumber
	\|f(x)\|_1 \le  \phi({\beta}) = \frac{\sqrt{n}\sum_{i=1}^{l}{\|}{w}_i^{\prime}{\|_1}|{t}_i|}{\lambda_{min}\Big{[}\sum_{i=1}^{l}{w}_i^{\prime}G_{\beta}\Big{(}\gamma \|w_i^{\prime}\|_1+|t_i|\Big{)}{w}_i\Big{]}}
	\label{phinew}
	\end{equation}
	where $\phi({\beta})$ is a continuous and monotonically decreasing function of $\beta$. It satisfies $\lim\limits_{\beta \rightarrow 0^{+}} \phi({\beta})= \infty$  and
	\begin{equation}\nonumber
	\lim\limits_{\beta \rightarrow \infty} \phi({\beta})= \xi = \frac{\sqrt{n}\sum_{i=1}^{l}{\|}{w}_i^{\prime}{\|_1}|{t}_i|}{\lambda_{min}\Big{[}\sum_{i=1}^{l}{w}_i^{\prime}{w}_i\Big{]}}.
	\end{equation}
	Therefore, if $\gamma> \xi$, the equation $\phi(\beta)=\gamma$ has a unique solution $\beta^{*}$ over $(0, \infty)$. Note that $G_{\beta_i}(e_i) \ge G_{\beta^{*}}(e_i)$ with $\beta_i \ge \beta^{*}$
	and $\lambda_{\min}[R_{ww}]=\lambda_{min}[\sum_{i=1}^{l}{w_i}^{\prime}G_{\beta_i}(e_i){w_i}]>0$ for any value of $\beta_i$ (this is the assumption). Hence one has $\lambda_{\min}\big[\sum_{i=1}^{l} {w}_i^{\prime} G_{\beta_i}\big{(}\gamma \|w_i^{\prime}\|_1+|t_i|\big{)}{w}_i\big] \ge \lambda_{\min}\big[\sum_{i=1}^{l} {w}_i^{\prime} G_{\beta}\big{(}\gamma \|w_i^{\prime}\|_1+|t_i|\big{)}{w}_i\big]$, and then
	$\bar{\phi}(\bar{\beta}) \le 	{\phi}({\beta}^{*})=\xi$
	with $\beta_i \ge \beta^{*}$ for $i=1,2,\cdots,\l$. Finally, we have $\|f(x)\|_1 \le \bar{\phi}(\bar{\beta}) \le 	{\phi}({\beta}^{*})=\xi \le \gamma$ for all $x \in \{x \in \mathbb{R}^{n}: \|x\|_1 \le \gamma\}$. This completes the proof.
\end{proof}

\subsection{Proof of Theorem \ref{theorem9}}
\label{proofTheorem9}
\begin{proof}
	By Theorem \ref{theorem8}, we have $f(x) \le \gamma$ if $\|x\|_1 \le \gamma$ and $\beta_i \ge \beta^{*}$ for all $i$. To prove $\| \nabla_{x} f(x) \|_1 \le \eta$, it is sufficient to prove $\forall j$, $\|\frac{\partial }{\partial x_j}\|_{1} \le \eta$. Based on the knowledge that $\frac{\partial \mathbf{U}^{-1}}{\partial \mathrm{x}}=-\mathbf{U}^{-1}\frac{\partial \mathbf{U}}{\partial \mathrm{x}}\mathbf{U}^{-1}$  and  $\frac{\partial \mathbf{U}\mathbf{V}}{\partial \mathrm{x}}= \frac{\partial \mathbf{U}}{\mathrm{x}}\mathbf{V}+\mathbf{U}\frac{\partial \mathbf{V}}{\mathrm{x}}$ where $\mathbf{U}$ and $\mathbf{V}$ are matrices and $\mathrm{x}$ is a scalar, we have 
	\begin{equation}
	\begin{small}
	\begin{aligned}
	&\frac{\partial }{\partial {x}_j} f(x)=\frac{\partial}{{x}_j}R_{ww}^{-1} P_{wt}\\
	&=-R_{ww}^{-1} \Big{[}\frac{\partial }{\partial {x}_j}R_{ww}\Big{]} R_{ww}^{-1} P_{wt} + R_{ww}^{-1}\Big{[}\frac{\partial }{{x}_j} P_{wt}\Big{]}\\
	&=-R_{ww}^{-1}\Big{[}\frac{\partial}{\partial {x}_j} \sum_{i=1}^{l}{w}_i^{\prime}G_{\beta_i}(e_i){w}_i\Big{]} f(x) + R_{ww}^{-1} \Big{[}\frac{\partial }{{x}_j} \sum_{i=1}^{l}{w}_i^{\prime}G_{\beta_i}(e_i){t}_i\Big{]}\\
	&= -R_{ww}^{-1}\Big{[} \sum_{i=1}^{l}{w}_i^{\prime}\Big(\frac{2e_i}{\beta_i^2}w_{i,j}G_{\beta_i}(e_i)\Big){w}_i\Big{]} f(x) \\&+ R_{ww}^{-1} \Big{[}\sum_{i=1}^{l}{w}_i^{\prime}\Big(\frac{2e_i}{\beta_i^2}w_{i,j}G_{\beta_i}(e_i)\Big){t}_i\Big{]}
	\label{gradient}
	\end{aligned}
	\end{small}
	\end{equation}
	where $w_{i,j}$ is the $j$-th element of $w_i$ and $x_j$ is $j$-th element of vector $x$. Taking one norm in both sides of \eqref{gradient}, we have 
	\begin{equation}
	\begin{aligned}
	&\|\frac{\partial }{\partial x_j} f(x) \|_{1}= \Big{\|}\Big{\{}-R_{ww}^{-1}\Big{[} \sum_{i=1}^{l}{w}_i^{\prime}\Big(\frac{2e_i}{\beta_i^2}w_{i,j}G_{\beta_i}(e_i)\Big){w}_i\Big{]} f(x) \\&+ R_{ww}^{-1} \Big{[}\sum_{i=1}^{l}{w}_i^{\prime}\Big(\frac{2e_i}{\beta_i^2}w_{i,j}G_{\beta_i}(e_i)\Big){t}_i\Big{]} \Big{\}}\Big{\|_1} \\
	& \le     
	\Big{\|}-R_{ww}^{-1}\Big{[} \sum_{i=1}^{l}{w}_i^{\prime}\Big(\frac{2e_i}{\beta_i^2}w_{i,j}G_{\beta_i}(e_i)\Big){w}_i\Big{]} f(x)\Big{\|_1} \\& + \Big{\|}R_{ww}^{-1} \Big{[}\sum_{i=1}^{l}{w}_i^{\prime}\Big(\frac{2e_i}{\beta_i^2}w_{i,j}G_{\beta_i}(e_i)\Big){t}_i\Big{]} \Big{\|_1}
	\label{gradient1}
	\end{aligned}
	\end{equation}
	Moreover, we have
	\begin{equation}
	\begin{aligned}
	&\Big{\|} -R_{ww}^{-1}\Big{[} \sum_{i=1}^{l}{w}_i^{\prime}\Big(\frac{2e_i}{\beta_i^2}w_{i,j}G_{\beta_i}(e_i)\Big){w}_i\Big{]} f(x) \Big{\|}_1 \\
	&\le 2\|R_{ww}^{-1}\|_1 \Big{\|}\Big{[} \sum_{i=1}^{l}{w}_i^{\prime}\Big(\frac{e_i}{\beta_i^2}w_{i,j}G_{\beta_i}(e_i)\Big){w}_i\Big{]} \Big{\|}_1 \|f(x)\|_1 \\
	&\overset{(\mathrm{IV})}{\le} 2\gamma\|R_{ww}^{-1}\|_1  \sum_{i=1}^{l}\Big{\|}{w}_i^{\prime}\Big(\frac{ e_i}{\beta_i^2}w_{i,j}G_{\beta_i}(e_i)\Big){w}_i \Big{\|}_1 \\
	&\overset{(\mathrm{V})}{\le} 2\gamma\|R_{ww}^{-1}\|_1  \sum_{i=1}^{l}\frac{|t_i|+\gamma\|w_i^{\prime}\|_{1}}{\beta_i^2}\|w_i^{\prime}\|_1\|w_i^{\prime}w_i\|_1 \\
	\label{gra_part1}
	\end{aligned}
	\end{equation}
	where (IV) comes from the convexity of vector $\ell_1$ norm and $f(x) \le \gamma$, (V) comes from $|e_iw_{i,j}| \le (|t_i|+\gamma \|w_i^{\prime}\|_1 )\|w_i^{\prime}\|_1$ and $G_{\beta_i}(e_i) \le 1$.
	Similarly, we have
	\begin{equation}
	\begin{aligned}
	&\Big{\|}R_{ww}^{-1} \Big{[}\sum_{i=1}^{l}{w}_i^{\prime}\Big(\frac{2e_i}{\beta_i^2}w_{i,j}G_{\beta_i}(e_i)\Big){t}_i\Big{]}
	\Big{\|}_1  \le  \\
	&2\|R_{ww}^{-1}\|_{1} \sum_{i=1}^{l} \frac{|t_i|+\gamma\|w_i^{\prime}\|_{1}}{\beta_i^2} \|w_i^{\prime}\|_1\|w_i^{\prime}t_i\|_1 
	\label{gra_part2}
	\end{aligned}
	\end{equation}
	Substituting \eqref{RWW}, \eqref{RWW_lam}, \eqref{gra_part1}, and \eqref{gra_part2} into \eqref{gradient1}, we obtain
	\begin{equation}
	\begin{aligned}
	&\|\frac{\partial }{\partial x_j} f(x) \|_{1} \le \bar{\psi}({\bar{\beta}})\\&= \frac{2\sqrt{n} \sum_{i=1}^{l}\frac{|t_i|+\gamma\|w_i^{\prime}\|_{1}}{\beta_i^2} \|w_i^{\prime}\|_1  \big{(}\gamma\|w_i^{\prime}w_i \|_1+ \|w_i^{\prime}t_i\|_1 \big{)} }{\lambda_{\min}\Big{[} 
		\sum_{i=1}^{l}{w}_i^{\prime}G_{\beta_i}\Big{(}\gamma \|w_i^{\prime}\|_1+|t_i|\Big{)}{w}_i\Big{]} }.
	\label{psibar}
	\end{aligned}
	\end{equation}
	If we set all kernel bandwidths to be the same with $\beta_i=\beta$ for all $i$, we arrive at 
	\begin{equation}
	\begin{aligned}
	&\|\frac{\partial }{\partial x_j} f(x) \|_{1} \le {\psi}({\beta})\\&
	= \frac{2\sqrt{n} \sum_{i=1}^{l}(|t_i|+\gamma\|w_i^{\prime}\|_{1}) \|w_i^{\prime}\|_1  \big{(}\gamma\|w_i^{\prime}w_i \|_1+ \|w_i^{\prime}t_i\|_1 \big{)} }{\beta^2\lambda_{\min}\Big{[} 
		\sum_{i=1}^{l}{w}_i^{\prime}G_{\beta}\Big{(}\gamma \|w_i^{\prime}\|_1+|t_i|\Big{)}{w}_i\Big{]} }.
	\label{psi}
	\end{aligned}
	\end{equation}
	One can see that \eqref{psi} is a continuous and monotonically decreasing function satisfying $\lim\limits_{\beta \to 0^{+}} \psi(\beta) = \infty$ and $\lim\limits_{\beta \to \infty} \psi(\beta) = 0$. This implies that $\psi(\beta)=\eta$ has a unique solution $\beta^{+}$ and $\psi(\beta) \le \eta$ if $
	\beta \ge
	\beta^{+}$. Observing \eqref{psibar} and \eqref{psi}, we have $\bar{\psi}(\bar{\beta}) \le \psi(\beta^{+})$ if $\beta_i \ge \beta^{+}$ for all $i$. This reveals that $0<\bar{\psi}(\bar{\beta}) \le \eta$ if $\forall i, \beta_i \ge \beta^{+}$. This completes the proof.
\end{proof}

\bibliographystyle{myIEEEtran}
\bibliography{GMKMCKF_twoColu_final}

\end{document}